\documentclass[lettersize,journal]{IEEEtran}

\usepackage{cite}
\usepackage{amsmath,amssymb,amsfonts}
\usepackage{bbm}
\usepackage{algorithmic}
\usepackage{algorithm}
\usepackage{float}
\usepackage{comment}
\usepackage[inline]{enumitem}
\usepackage[usenames,dvipsnames]{xcolor}
\usepackage{textcomp}
\usepackage{amsthm}
\usepackage{pifont}
\usepackage{tikz}
\usepackage{mathtools}
\usetikzlibrary{shapes.geometric, arrows}

\tikzstyle{startstop} = [rectangle, rounded corners, 
minimum width=2cm, 
minimum height=1cm,
text centered, 
draw=black, 
thick,
text width=3.5cm
]

\tikzstyle{io} = [trapezium, 
trapezium stretches=true, 
trapezium left angle=70, 
trapezium right angle=110, 
minimum width=3cm, 
minimum height=1cm, text centered, 
draw=black, fill=blue!30, ultra thick]

\tikzstyle{process} = [rectangle, 
minimum width=3cm, 
minimum height=1cm, 
text centered, 
text width=3cm, 
draw=black, ultra thick]

\tikzstyle{decision} = [diamond, 
minimum width=3cm, 
minimum height=1cm, 
text centered, 
draw=black, ultra thick]
\tikzstyle{arrow} = [ultra thick,->,>=stealth]
\newtheorem{thm}{Theorem}
\newtheorem{cor}[thm]{Corollary}
\newtheorem{lem}{Lemma}

\newtheorem{defn}{Definition}

\newtheorem{rem}{Remark}
\newtheorem{prob}{Problem}

\begin{document}
\title{Topology Reconstruction of a Class of Electrical Networks with Limited Boundary Measurements}
\author{Shivanagouda Biradar \& Deepak U Patil
\thanks{S. Biradar \& D. U. Patil are with the Department of
	Electrical Engineering, Indian Institute of Technology Delhi, India. e-mail: shivmb59@gmail.com, deepakpatil@ee.iitd.ac.in}
\thanks{}}
\markboth{Journal of \LaTeX\ Class Files,~Vol.~14, No.~8, August~2021}%
{Shell \MakeLowercase{\textit{et al.}}: A Sample Article Using IEEEtran.cls for IEEE Journals}

\maketitle
\begin{abstract}
	We consider the problem of recovering the topology and the edge conductance value, as well as characterizing a set of electrical networks that satisfy the limitedly available Thevenin impedance measurements. The measurements are obtained from an unknown electrical network, which is assumed to belong to a class of circular planar passive electrical network. This class of electrical networks consists of R, RL, and RC networks whose edge impedance values are equal, and the absolute value of the real and the imaginary part of the edge impedances are also equal. To solve the topology reconstruction and the set characterization problem, we establish a simple relation between Thevenin impedance and the Laplacian matrix and leverage this relation to get a system of multivariate polynomial equations, whose solution is a set of all electrical networks satisfying the limited available Thevenin's impedance measurements. To confine the search space and generate valid electrical networks, we impose the triangle and Kalmanson's inequality as constraints. The solution to a constrained system of multivariate polynomial equations is a set of reconstructed valid electrical networks. For simple algorithmic solutions, we use Gr\"obner basis polynomials. This paper shows that the triangle and the Kalmanson's inequality holds for general circular planar passive R, RL, and RC electrical networks if certain boundary conditions lie within a convex cone. Numerical examples illustrate the developed topology reconstruction method.
\end{abstract}
\begin{IEEEkeywords}
	Topology reconstruction, Gr\"{o}bner basis, electrical network.
\end{IEEEkeywords}
\section{Introduction}
\label{sec:introduction}
Electrical networks are an integral part of our lives. They are used in innumerable ways to make our lives easier. Resistor networks are used in geology to model fractures in crystalline rock \cite{jirkuu2019resistor}. Electrical resistivity studies of carbon composites are carried out by modeling them as a 3D resistor network model \cite{lundstrom2022resistor}. Water distribution networks are modeled as resistor networks for fault detection \cite{veldman2015towards}. Resistor networks are also used in applications like soft robotics sensor array \cite{zhao2023measuring}, modeling graphene sheet and carbon nanotube \cite{cheianov2007random}. Similarly, resistor-inductor and resistor-capacitor networks are used to model and study elastic arterial vessels \cite{silva2019schsim}, piezoelectric based vibration absorber \cite{zhao2015dynamic}, rock failure process under high-voltage electropulse \cite{rao2022mechanism}, ceramic coaxial resonator \cite{nguenouho2017ceramic}. For studying and analyzing such critical systems, the network structure and the corresponding edge conductances must be known apriori, which is usually unavailable in many practical cases. Therefore, topology reconstruction is essential. The objective of topology reconstruction for electrical networks is $1)$ to find the unknown network structure and $2)$ to compute the edge conductance values using available electrical measurements.

The topology reconstruction problem for electrical networks is a widely researched topic. In \cite{curtis1994finding}, authors consider a class of resistor networks denoted as $C(m,n)$, $m$ is the number of circles arranged concentrically, and $n$ represents the number of rays emanating from $n$ boundary nodes positioned on the outermost circle. Therefore, the structure is presumed to be known. Additionally, it is assumed that all boundary terminals are accessible for taking measurements. Next, the unknown edge resistance values are calculated using the boundary measurement. In \cite{curtis1990determining}, authors assume that networks take a general rectangular mesh-type structure, with boundary nodes available for boundary measurements. The authors propose an algorithm that uses a gamma harmonic function $\left(\mbox{based on Kirchhoff's law}\right)$ to compute the edge resistance values. A similar problem of network topology reconstruction has been studied widely in phylogenetics, wherein genetic distance measure, akin to resistance distance, is used to reconstruct the phylogenetic network \cite{forcey2020phylogenetic}. The authors in \cite{calzavara2021structured} examine the reconstruction of a class of RC network utilizing input-output data. The input-output data is first used to identify a state space model, and then a transformation is computed that reveals the RC network structure from the state space model. 

\noindent In our work in \cite{biradar2023topology}, we consider an unknown general circular planar resistor network to be reconstructed, whose response matrix is assumed to be known. We develop an algorithm that uses the response matrix to enumerate all electrical networks satisfying the response matrix. In \cite{biradar2022topology}, we consider the problem of characterizing a set of resistive networks corresponding to partially available resistance distance measurements.

However, the network structure is rarely known apriori, and not all boundary nodes are available for conducting experiments and collecting measurements. For example, in a soft resistive sensor array network, there is no prior knowledge of the structural information, and only a few boundary terminals are accessible for collecting boundary measurements. It is also observed that for partial measurements, multiple network configurations are possible, which are not accounted for in general. The knowledge of multiple network configurations for available electrical measurements is essential in studying fault detection and robustness, disturbance propagation, integrated circuit design, circuit debugging, and designing power management strategies in large-scale power networks.

This paper characterize electrical networks that satisfy the partially available  Thevenin's impedance measurements using the Gr\"{o}bner basis. These measurements come from a class of electrical networks, which consists of resistive (R), resistor-inductor (RL), and resistor-capacitor (RC) networks whose edge impedances are equal, and the absolute values of real and imaginary parts are also equal. Such class of networks are commonly found in lattice models of crystal structure in physical and material sciences, a sensor network designed to measure physical quantities like temperature, force, and irradiance, as well as in antenna matching and RL and RC filters. Such resistive networks are also utilized to study the properties of simple, connected, undirected graphs.

Major contributions of this paper are,
\begin{itemize}
	\item In contrast to works in \cite{forcey2020phylogenetic}, \cite{curtis1994finding}, \cite{ghosh2008minimizing}, which assume specific network structures, our approach relaxes these assumptions
	and requires only a broader assumption of circular planarity.  
	
	\item We establish that the triangle inequality and Kalmanson's inequality hold for general circular planar passive RL and RC networks, provided that specific boundary conditions lie within a particular convex cone. Furthermore, we demonstrate that the determinant of appropriately chosen sub-matrices of the Laplacian matrix satisfies the triangle and the Kalmanson's inequality for a class of electrical networks. 
	
	\item For the topology reconstruction and set characterization, we develop an algorithm wherein each Thevenin's impedance measurement is related to an unknown Laplacian matrix, leading to a system of multivariate polynomial equations. The system of polynomial equations, along with the triangle inequality, the \emph{valid} Kalmanson's inequality, and graph connectedness constraint, are solved to obtain a \emph{set of valid electrical networks satisfying the available measurements}.
	
	\item Imposing the valid Kalmanson's inequality necessitates knowledge of specific boundary conditions, which may be inaccessible due to the unavailability of certain boundary terminals. To address this challenge, we develop an algorithm that infers appropriate boundary conditions and selects the corresponding Kalmanson's inequality based on limited boundary measurements. This approach effectively reveals the internal characteristics of the unknown electrical network.
	
	\item Solving a system of high-degree multivariate polynomial equations is computationally demanding. Therefore, to make computations tractable, we use an algebraic geometry-based technique, i.e., \emph{Gr\"{o}bner basis method}.
	
	\item For the class of electrical networks defined above, we establish a relation between the Thevenin's impedance and the corresponding Laplacian matrix. This relation is an extension of result on computing the resistance distance of an undirected graph \cite{babic2002resistance}.
\end{itemize}
\noindent \textit{\textbf{Mathematical Notations:}} The set $\mathbb{Z}$, $\mathbb{Z}^+$, $\mathbb{C}$, $\mathbb{R}^{+}$ are set of integers, positive integers, complex and positive real numbers, respectively. $\mathbf{1}$ is a vector of ones. $\mathcal{L}[i]$ is a submatrix of Laplacian matrix obtained by deleting $i^{th}$ row and $i^{th}$ column, whereas $\mathcal{L}[i,j]$ is a submatrix of Laplacian matrix obtained by deleting $i^{th}$ \& $j^{th}$ row, and $i^{th}$ \& $j^{th}$ column. The submatrix $\mathcal{L}[S,T]$, of the Laplacian matrix $\mathcal{L}$, is obtained by excluding rows and columns indexed by elements in set $S$ and $T$, whereas submatrix $\mathcal{L}(S,T)$ is obtained by including elements from set $S$ and $T$. $|\cdot|$ is the cardinality of set. A set $\{n-i\}=\{1,\ldots,n-i\}, \forall i<n$. Bold, small letters are used for representing vectors, like $\mathbf{v}$, and bold, capital letters are used for representing matrix, for example $\mathbf{A}$. A vector $\mathbf{v}\left(\{n-i\}\right)$ is sub vector of $\mathbf{v}$ obtained by including elements of $\mathbf{v}$ indexed by set $\{n-i\}$. $\mathbf{e}_i$ is a canonical vector whose $i^{th}$ element is $1$ and all other elements are $0$. Then, $\mathbf{e}_{ij}=\mathbf{e}_i-\mathbf{e}_j$. Consider a complex number $z=x+\mathbf{j}y$, $\Re(z)=x$ is the real part, and $\Im(z)=y$ is the imaginary part of the complex number $z$. $\odot$ is a binary operation and represents the element-wise product.
\section{Preliminaries}
\label{prelims}
\subsection{Network Structure}
\label{netstruct}
Consider a circular planar passive electrical network $\Gamma=(\mathcal{G},\gamma)$. A finite, simple, circular planar graph $\mathcal{G}=(\mathcal{V_B}, \mathcal{E})$ is a graph embedded in disc D, on the plane, bounded by a circle C.
The set $\mathcal{V_B}$, is the set of boundary nodes and $\mathcal{E} \subseteq \mathcal{V_B} \times \mathcal{V_B}$ is the set of edges.  The boundary nodes, labeled from $1$ to $n$, lie on C in clockwise circular order. The boundary nodes say $j,k,l,m$, are said to be in circular order if the ordering $j<k<l<m$ is induced by the angle of arc measured in a clockwise direction from $j$. We partition the set $\mathcal{V_B}$ based on the availability of boundary nodes for voltage and current measurements $(\mbox{boundary measurements})$. Therefore, let $\mathcal{A}$ be the set of boundary nodes available for boundary measurements and $\mathcal{U_B}=\mathcal{V_B} \setminus \mathcal{A}$ be a set of unavailable boundary nodes. The conductance function $\gamma:\mathcal{E} \rightarrow \mathbb{C}$, assigns to each edge $\sigma$ a complex value $\gamma(\sigma)$, known as the conductance of $\sigma$. The impedance of each edge $\sigma$ is $z\left(\sigma\right)=\gamma \left( \sigma \right)^{-1}$.

In this paper, we consider that the network $\Gamma$ is unknown and is to be reconstructed. We assume that this unknown network $\Gamma$ is either a R, RL, or RC network with no voltage source. Also, every edge in an unknown network takes the edge impedance value of the form,
\begin{enumerate*}
	\item $z(\sigma)=\beta$ for R network,   
	\item $z(\sigma)=\beta+\mathbf{j}\beta$, for RL network,  
	\item $z(\sigma)=\beta-\mathbf{j}\beta$, for RC network, 
\end{enumerate*}
$\forall \sigma \in \mathcal{E}$ and $\beta \in \mathbb{R}^+$. Let us label $\Gamma_{\beta}$ as a set of circular planar passive R, RL, and RC electrical networks possessing the above-described properties. In class $\Gamma_\beta$, let $\Gamma_{\beta}^R \subset \Gamma_\beta$ be a set of resistive networks, $\Gamma_{\beta}^{RL}\subset \Gamma_\beta$ is a set of RL network and $\Gamma_{\beta}^{RC}\subset \Gamma_\beta$ is a set of RC network. Let $\Upsilon$ be a collection of all general R, RL, and RC networks, such that $\Gamma_{{\beta}} \subset \Upsilon$.  Let $\Upsilon^{R}$, $\Upsilon^{RL}$ and $\Upsilon^{RC} \subset \Upsilon$ be a set of all general R, RL, and RC networks, respectively.

The relation between the  Thevenin's impedance $z^{th}_{i,j} \in \mathbb{C}$ across boundary nodes $i,j \in \mathcal{V_B}$, and the Laplacian matrix takes a special form, which will be used to construct a set of multivariate polynomials. It will eventually be used in topology reconstruction and set characterization problems. The relation is stated and proved in the next section.
\subsection{ Thevenin's Impedance \& Laplacian Matrix} 
\label{laplace_impedance_relation}
The  Thevenin's impedance, $z^{th}_{i,j} \in \mathbb{C}$, between any two boundary nodes $i,j \in \mathcal{V_B}$, is the effective impedance measured across boundary nodes $i$ and $j$. The effective impedance is the ratio of applied voltage to the response current. To facilitate the topology reconstruction and set characterization, we measure  Thevenin's impedances $z^{th}_{s,t}$ across all the available boundary nodes $s,t \in \mathcal{A}$. Then, a set of available  Thevenin's impedance measurements is a set $Z^{th}=\{z^{th}_{s,t}=v_{s,t}/i_s : \forall \, s,t \in \mathcal{A} \}$. 
The Laplacian matrix $\mathcal{L}(\mathcal{G})$, of a graph $\mathcal{G}$, is in general defined as follows:
\begin{equation}\label{lap_def}
	\left[\mathcal{L}\left(\mathcal{G}\right)\right]_{ij}=\left[\mathcal{L}_{\mathcal{G}}(i,j)\right]\left\{
	\begin{array}{ll}
		= -1, & \mbox{if $ij\in \mathcal{E}$},\\
		= d_i, & \mbox{if $i=j$},\\
		=0, & \mbox{otherwise}.
	\end{array}
	\right.
\end{equation}   
\noindent Where $d_i$ is the degree of node $i$. The Laplacian matrix of any electrical network $\Gamma \in \Gamma_{\beta}$, $\mathcal{L}_\Gamma$, admits the following form, i.e., $\mathcal{L}_\Gamma=\gamma \mathcal{L(G)}$, where $\gamma=\frac{1}{z}$. Since, $\Gamma \in \Gamma_{\beta}$ is unknown, the corresponding Laplacian matrix  $\mathcal{L}_{\Gamma}$ is also unknown. We, therefore, aim to characterize a set of $\mathcal{L}_{\Gamma}$, which satisfies measurements in set $Z^{th}$, thereby also reconstructing the unknown electrical network $\Gamma$. To achieve this, we first establish a simple relation between Thevenin's impedance and the Laplacian matrix $\mathcal{L}_{\Gamma}$, elucidated in Theorem \ref{impedance_dist}.   
\begin{thm}
	\label{impedance_dist}
	Consider a circular planar passive electrical network $\Gamma=\big(\mathcal{G},\gamma \big)$, where $\Gamma \in \Gamma_{\beta}$. Let $\mathcal{G}=\big(\mathcal{V_B},\mathcal{E}\big)$ be a simple, connected graph on $n$ vertices, where $n\ge3$, corresponding to $\Gamma$. Then, the  Thevenin's impedance $z^{th}_{j,k}$ across any pair of boundary nodes $j$ and $k$, $\forall$ $1\le j \ne k \le n_b$, is given as:
	\begin{equation}\label{imped_dist}
		{z^{th}_{j,k}} = \gamma^{-1} \frac{{\det{\mathcal{L_G}\left[ {j,k} \right]} }}{{\det  {\mathcal{L_G}\left[ k,k \right]} }}.
	\end{equation} 
\end{thm}
\begin{proof}
	The  Thevenin's impedance $z^{th}_{j,k}$ is an equivalent impedance as seen through boundary nodes $j$ and $k$, and is computed as the ratio of applied voltage, i.e., $1\angle0$,  to the response current $\mathbf{i}_{j}$, shown in Fig.\ref{fig1:zthcalc},
	\begin{figure}[h]
		\centering
		\includegraphics[scale=0.25]{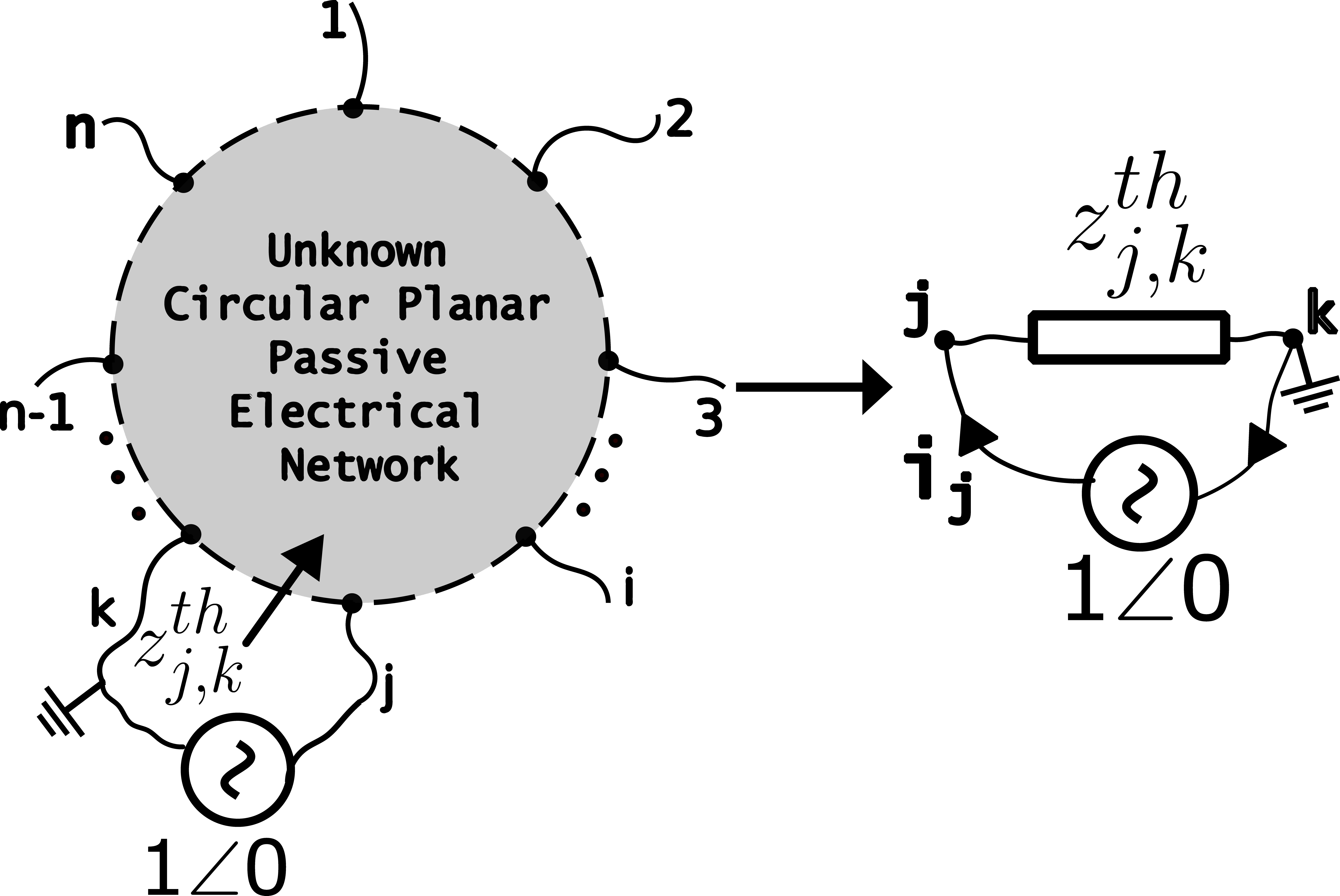}
		\caption{Computing  Thevenin's impedance $z^{th}_{j,k}$.}
		\label{fig1:zthcalc}
	\end{figure}
	\noindent The response current $\mathbf{i}_j$ is given as,
	\begin{equation}
		\mathbf{i}_j=\sum_{p\in\mathcal{N}\left(j\right)}{(v_j-v_p)\gamma{(jp)}},
	\end{equation}
	where $\mathcal{N}(j)$ is a set of nodes in neighbourhood of node $j$.
	The  Thevenin's impedance $z^{th}_{j,k}$ is then defined as,
	\begin{equation}
		z^{th}_{j,k}=\left(\mathbf{i}_j\right)^{-1}=\gamma^{-1}\left(d_j-\sum_{p\in\mathcal{N}(j)}{v_p}\right)^{{-1}},
	\end{equation}
	where $d_j$ is the degree of $j^{th}$ node. 
	Without loss of generality let $j=n-1$ and $k=n$, then the  Thevenin's impedance $z^{th}_{(n-1),n}$ is,
	\begin{equation}\label{eqn:thev_1}
		z^{th}_{(n-1),n}=\gamma^{-1}\left(d_{n-1}-\sum_{p\in\mathcal{N}(n-1)}{v_p}\right)^{-1}.
	\end{equation}
	From Kirchhoff's current law, we have,
	\begin{equation}\label{eqn:kcl}
		\sum_{m \in \mathcal{N}(l)}(v_l-v_m)\gamma(lm)=0,\forall l \ne \left(n-1\right),n.
	\end{equation}
	Equation (\ref{eqn:kcl}) can be rewritten as,
	\begin{equation}\label{thm1:eqn8}
		\mathcal{L}_{\Gamma}\big(\{n-2\},\{n-2\}\big)\mathbf{v}\big(\{n-2\}\big)=\mathbf{b}_{\Gamma},
	\end{equation}
	since $v_{n-1}=1\angle 0$ and $v_n=0$. Here $\mathbf{v}{\left(\{n-2\}\right)}=\left[v_1,v_2,\ldots,v_{(n-2)}\right]^T$, $\mathcal{L}_{\Gamma}\big(\{n-2\},\{n-2\}\big)$ is a submatrix obtained by including elements indexed by set $\{n-2\}$,  $\mathbf{b}_{\Gamma}=\left[b_{\Gamma,1}, b_{\Gamma,2},\ldots,b_{\Gamma,(n-2)}\right]^\mathbf{T}$,  and $b_{\Gamma,k}=\gamma(k(n-1))$ if nodes $k$ and $(n-1)$ are incident, otherwise $b_{\Gamma,k}=0$. Let us denote the submatrix $\mathcal{L}_{\Gamma}\left(\{n-2\},\{n-2\}\right)$ as $\mathcal{L}_{\Gamma,\{n-2\}}$ and $\mathbf{v}\left(\{n-2\}\right)$ as $\mathbf{v}_{\{n-2\}}$, then equation (\ref{thm1:eqn8}) is, 
	\begin{equation}\label{eqn:ln-2}
		\mathcal{L}_{\Gamma,\{n-2\}}\mathbf{v}_{\{n-2\}}=\mathbf{b}_{\Gamma}.
	\end{equation}
	Using the fact that all the edge impedances are equal, equation (\ref{eqn:ln-2}) can be rewritten as, 
	\begin{equation}\label{ln-3}
		\mathcal{L}_{\mathcal{G},\{n-2\}}\mathbf{v}_{\{n-2\}}=\mathbf{b}_{\mathcal{G}},   
	\end{equation}
	where $\mathbf{b}_{\mathcal{G}}=\left[b_{\mathcal{G},1}, b_{\mathcal{G},2},\ldots,b_{\mathcal{G},(n-2)}\right]^\mathbf{T}$, $b_{\mathcal{G},k}=1$ if nodes $k$ and $(n-1)$ are incident, otherwise $b_k=0$.
	\noindent From equation (\ref{ln-3}), the node voltages $\mathbf{v}_{\{n-2\}}$ is given as,
	\begin{equation}
		\mathbf{v}_{\{n-2\}}=\mathcal{L}_{\mathcal{G},\{n-2\}}^{-1}\mathbf{b}_{\mathcal{G}}.	
	\end{equation}	
	Therefore, node voltage at $p^{th}$ node is given as,
	\begin{equation}\label{eqn:thev2}
		v_p = \sum_{l \in \mathcal{N}(n-1)} t_{pl},
	\end{equation}
	where $t_{pl}$ is $(p,l)^{th}$ element of $\mathcal{L}_{\mathcal{G},\{n-2\}}^{-1}$.

	\noindent Let us compute value of $t_{pl}$. We know that,
	\begin{equation}
		\mathcal{L}_{\mathcal{G},\left\{n - 2\right\}}^{ - 1} = \frac{{adj\left( {{\mathcal{L}_{\mathcal{G},\{n - 2\}}}} \right)}}{{\det{\mathcal{L}_{\mathcal{G},\{n - 2\}}}}}.
	\end{equation}
	Here, $adj$ is the adjoint of a matrix. Then, $t_{pl}$ is,
	\begin{equation}
		{t_{pl}} = {\left( { - 1} \right)^{p + l}}\frac{{\det {\mathcal{L}_{\mathcal{G},\{n - 2\}}}\left[{p,l}\right] }}{{\det {\mathcal{L}_{\mathcal{G},\{n - 2\}}}}}.
	\end{equation}
	Here ${\mathcal{L}_{\mathcal{G},\{n - 2\}}}\left[{p,l}\right]$ is a submatrix obtained by deleting $p^{th}$ row and $l^{th}$ column. Therefore, the node voltage $v_p$ is,  
	\begin{equation}
		{v_p} = \sum\limits_{p \in N\left( {n - 1} \right)} {{{\left( { - 1} \right)}^{p + l}}\frac{{\det {L_{\mathcal{G},\left\{ {n - 2} \right\}}}\left[ {p,l} \right]}}{{\det {L_{\mathcal{G},\left\{ {n - 2} \right\}}}}}}. 	
	\end{equation}
	\noindent Substitute value of $v_p$ in equation (\ref{eqn:thev2}) into equation (\ref{eqn:thev_1}), therefore we have, 
	\begin{equation}\label{eqn:thev_3}
		\begin{split}
			z^{th}_{(n-1),n}&=\gamma^{-1}\left(d_{n-1}-\sum_{p\in\mathcal{N}(n-1)}{\sum_{l \in \mathcal{N}(n-1)} t_{pl}}\right)^{-1},\\
			&= \gamma^{-1} I^{-1},
		\end{split}
	\end{equation}
	and 
	\begin{equation}\label{eqn:thm1_17}
		\begin{split}
			I=&\left(d_{n-1}-\sum_{p\in\mathcal{N}(n-1)}{\sum_{l \in \mathcal{N}(n-1)} t_{pl}}\right),\\
			=&\left(d_{n-1}-\sum_{p\in\mathcal{N}(n-1)}{\sum_{l \in \mathcal{N}(n-1)} {\left( {-1} \right)^{p + l}}\frac{{\det {\mathcal{L}_{\mathcal{G},\{n - 2\}}}\left[{p,l}\right] }}{{\det {\mathcal{L}_{\mathcal{G},\{n - 2\}}}}}}\right).
		\end{split}
	\end{equation}
	Let us look at $\det\mathcal{L}_{\mathcal{G},\{n-1\}}=\det\mathcal{L}_{\mathcal{G}}\left(\{n-1\},\{n-1\}\right)$, and expand it with respect to last column, then, 
	\begin{multline}\label{eqn:thm1_18}
		{\det {\mathcal{L}_{\mathcal{G},\{n - 1\}}}}=d_{n-1}{\det {\mathcal{L}_{\mathcal{G},\{n - 2\}}}}-\\
		\sum_{p\in \mathcal{N}(n-1)} (-1)^{p+(n-1)}{\det {\mathcal{L}_{\mathcal{G},\{n - 1\}}}}\left[p,n-1\right]
	\end{multline}
	The submatrix ${\mathcal{L}_{\mathcal{G},\{n - 1\}}}\left[p,n-1\right]$ is obtained by deleting $p^{th}$ row and $(n-1)^{th}$ column from ${\mathcal{L}_{\mathcal{G},\{n - 1\}}}$.
	The term ${\det {\mathcal{L}_{\mathcal{G},\{n - 1\}}}}\left[p,n-1\right]$ can be further expanded, with respect to the last row as,
	\begin{multline}\label{eqn:thm1_19}
		{\det {\mathcal{L}_{\mathcal{G},\{n - 1\}}}}\left[p,n-1\right]=\\\sum_{l \in \mathcal{N}(n-1)}(-1)^{l+(n-1)}{\det {\mathcal{L}_{\mathcal{G},\{n - 2\}}}}[p,l].
	\end{multline}
	Substitute equation (\ref{eqn:thm1_19}) in equation (\ref{eqn:thm1_18}), we get,
	\begin{multline}\label{eqn:thm1_20}
		{\det {\mathcal{L}_{\mathcal{G},\{n - 1\}}}}=d_{n-1}{\det {\mathcal{L}_{\mathcal{G},\{n - 2\}}}}-\\
		\sum_{p\in \mathcal{N}(n-1)}\sum_{l\in \mathcal{N}(n-1)} (-1)^{p+l}{\det {\mathcal{L}_{\mathcal{G},\{n - 2\}}}}\left[p,l\right].
	\end{multline}
	Now, substitute equation (\ref{eqn:thm1_20}) in equation (\ref{eqn:thm1_17}), we get,
	\begin{equation}\label{eqn:thm1_21}
		\begin{split}
			I = \frac{\det \mathcal{L}_{\mathcal{G},\{n-1\}}}{\det \mathcal{L}_{\mathcal{G},\{n-2\}}}
			=\frac{\det \mathcal{L}_{\mathcal{G}}[n,n]}{\det \mathcal{L}_{\mathcal{G}}[n-1,n]}.
		\end{split}
	\end{equation}
	To get $z^{th}_{(n-1),n}$, substitute equation (\ref{eqn:thm1_21}) into equation (\ref{eqn:thev_3}), we therefore have, 
	\begin{equation}
		\begin{split}
			z^{th}_{(n-1),n}= \gamma^{-1} I^{-1}
			=\gamma^{-1}\frac{\det \mathcal{L}_{\mathcal{G}}[n-1,n]}{\det\mathcal{L}_{\mathcal{G}}[n,n]}.
		\end{split}
	\end{equation}
	Similarly, the  Thevenin's impedance across nodes $j$ and $k$ in a circular planar passive electrical network $\Gamma \in \Gamma_{\beta}$, $z^{th}_{j,k}$ is given by,
	\begin{equation}\label{resist_form}
		z^{th}_{j,k}=\gamma^{-1}\frac{\det \mathcal{L}_{\mathcal{G}}[j,k]}{\det\mathcal{L}_{\mathcal{G}}[k,k]},
	\end{equation} 
	here $\det\mathcal{L}_{\mathcal{G}}[k,k]=\mathfrak{T}\left(\mathcal{G}\right),\,\forall k \in \mathcal{V_B}$, where $\mathfrak{T}\left(\mathcal{G}\right)$ is the number of spanning tree. \quad \quad \quad \quad \quad \quad \quad \quad \quad \quad \quad \quad 
\end{proof}
\noindent The  Thevenin's resistance $r^{th}_{j,k}$ across boundary nodes $j,k\in\mathcal{V_{B}}$, defined on networks in $\Gamma^{R}_{\beta} \subset \Gamma_{\beta}$ is, 
$$ z^{th}_{j,k}=r^{th}_{j,k}=\beta \frac{\det \mathcal{L}_{\mathcal{G}}[j,k]}{\det\mathcal{L}_{\mathcal{G}}[k,k]}.$$
Thevenin's resistance $r^{th}_{j,k},\,\forall j,k \in \mathcal{V_B}$ satisfies two important properties, i.e., the triangle and the Kalmanson's inequality. Both inequalities are stated in Theorem \ref{thm:tri_ineq1} and \ref{thm:kalmanson1}. 
\begin{thm}\big(Triangle Inequality\big)\label{thm:tri_ineq1}
	\cite{choi2019resistance} Consider a circular planar passive resistive network $\Gamma=\left(\mathcal{G}, \mathcal{V_{B}} \right) \in \Upsilon^{R}$. For any three distinct circularly ordered boundary nodes $j,k,l \in \mathcal{V_{B}}$, the  Thevenin's resistance $r_{j,k}^{th}$, $r_{k,l}^{th}$ and $r_{j,l}^{th}$  satisfies,
	\begin{equation*}\label{triineq}
		r_{j,l}^{th} \le r_{j,k}^{th}+r_{k,l}^{th}.
	\end{equation*}
\end{thm}
\begin{thm}\big(Kalmanson's  Inequality\big)\label{thm:kalmanson1}
	\cite{demidenko1997note} For any resistive network $\Gamma=\left(\mathcal{G}, \mathcal{V_{B}} \right) \in \Upsilon^{R}$, the Thevenin's resistances $r^{th}_{j,k}, r^{th}_{k,l}, r^{th}_{l,m}, r^{th}_{j,l}, r^{th}_{j,k}, \mbox{ and }r^{th}_{i,l}$ defined on any four distinct circularly ordered boundary nodes $j,k,l,m \in \mathcal{V_{B}}$ satisfies,
	\begin{equation*}
		\begin{array}{l}
			{r^{th}_{j,l}} + {r^{th}_{k,m}} \ge {r^{th}_{j,k}} + {r^{th}_{l,m}}\, \& \,
			{r^{th}_{j,l}} + {r^{th}_{k,m}} \ge {r^{th}_{k,l}} + {r^{th}_{j,m}}.
		\end{array}
	\end{equation*}
\end{thm}

\noindent In the following section, we prove that the triangle and the Kalmanson's inequalities on networks in $\Upsilon^{RL}$ and $\Upsilon^{RC}$ hold if specific boundary conditions are satisfied and Thevenin's impedance lies within a certain convex cone. \emph{These generalized inequalities will serve as constraints in a system of multivariate polynomial equations, whose solution yields a reconstructed topology and a set of valid electrical networks corresponding to measurements in $Z^{th}(\Gamma)$. Details of this reconstruction process are provided in section \ref{sec:prob_form}}.
\subsection{Triangle \& Kalmanson's Inequality on Networks in $\Upsilon$}\label{sec:triandKalman}
The  Thevenin's impedance defined on RL and RC networks is complex numbers; therefore, to define inequalities on  Thevenin's impedances, there must be an ordering, described below. 
\subsubsection{Ordering} 
Let $\left(x,y\right)$ be an ordered pair and, let $\mathcal{Q}_r$ be the set of all ordered pairs in $r^{th}$ quadrant, $\forall r \in \{1,2,3,4\}$. Thus $\mathcal{Q}_r$ is a cone. Now, consider ordered pairs $\left(x_a,y_a\right)\, \mbox{and} \, \left(x_b,y_b\right)$ in $\mathcal{Q}_r$. A binary relation $\preceq^\mathbf{r}$ on set $\mathcal{Q}_r$ is defined as  $\left(x_a,y_a\right) \preceq^\mathbf{r} \left(x_b,y_b\right)$ iff $\left(x_a,y_a\right)-\left(x_b,y_b\right)\in\mathcal{Q}_r$. Therefore,  $\left(\mathcal{Q}_r,\preceq^{\mathbf{r}}\right), \forall r \in \{1,2,3,4\}$ is a partially ordered set.
\subsubsection{Triangle Inequality}
Generalized triangle inequality for networks in $\Upsilon$ is given in Theorem \ref{thm:tri_ineq_RL_RC1}, with a detailed proof in Appendix \ref{App:A}.
\begin{thm}\label{thm:tri_ineq_RL_RC1}
	Let $\Gamma=\left(\mathcal{G},\gamma \right) \in \Upsilon$ be a circular planar passive electrical network. For any three distinct circularly ordered boundary nodes $j, k, l \in \mathcal{V_B}$, let $t={e_{jk}^T\mathcal{L}_\Gamma ^{\dagger}{e_{kl}}}$, where $t\in \mathbb{C} $, be the potential difference across boundary nodes $j$ and $k$ when $1\angle0 $ Ampere current is fed into $k^{th}$ node and $l^{th}$ node is grounded. Then, depending on the quadrant in which $t$ lies in a complex plane, the following inequalities hold,
	\begin{itemize}
		\item {If $\big(\Re \left( t \right),\Im \left( t \right)\big)\in\mathcal{Q}_{1}, \implies z^{th}_{j,k}+z^{th}_{k,l}\preceq^{\mathbf{r}}z^{th}_{j,l}$\vspace{0.1cm}}.
		\item {If $\big( \Re \left( t \right),\,\Im \left( t\right) \big)\in \mathcal{Q}_{3},\implies z^{th}_{j,l} \preceq^{\mathbf{r}} z^{th}_{j,k}+z^{th}_{k,l}.$}
		\item {If $\big( \Re \left( t \right),\,\Im \left( t \right) \big)\in \mathcal{Q}_{2},\\ \implies  \left(r^{th}_{j,l}, x^{th}_{j,k}+x^{th}_{k,l}\right)\preceq^\mathbf{r} \left(r^{th}_{j,k}+r^{th}_{k,l},x^{th}_{j,l}\right).$\vspace{0.1cm}}
		\item {If $ \big(\Re \left( t \right),\Im\left( t \right) \big)\in\mathcal{Q}_{4},\\  \implies \left(r^{th}_{j,k}+r^{th}_{k,l}, x^{th}_{j,l}\right)\preceq^\mathbf{r} \left(r^{th}_{j,l},x^{th}_{j,k}+x^{th}_{k,l}\right)$}\vspace{0.1cm}
	\end{itemize}
	Where, $\mathbf{r}=1$ for RL network, $\mathbf{r}=2$ for RC network.  
\end{thm}
\noindent The triangle inequality defined on networks in $\Gamma_{\beta}$ takes a special form as stated in Corrollary \ref{cor6_1}, and a detailed proof is given in Appendix \ref{App:A1}.  
\begin{cor}\label{cor6_1}
	For any electrical network $\Gamma=\left(\mathcal{G},\gamma \right)$, where $\Gamma \in \Gamma_{\beta}$. Choose any three distinct circularly ordered boundary nodes $j, k, l \in \mathcal{V_B}$. Then, the following inequality holds,
	\begin{equation}
		\det\mathcal{L_G}\left[j,l\right] \le \det\mathcal{L_G}\left[j,k\right] + \det\mathcal{L_G}\left[k,l\right].
	\end{equation}
\end{cor}
\subsubsection{Kalmanson's Inequality}
The Kalmanson's inequality on networks in $\Upsilon$ is stated below in Theorem \ref{thm:kalman_ineq_RL_RC21}. A detailed proof is given in Appendix \ref{App:B}.
\begin{thm}\label{thm:kalman_ineq_RL_RC21}
	Consider a circular planar passive electrical network $\Gamma=\left(\mathcal{G},\gamma \right)$, where $\Gamma \in \Upsilon$. For any four distinct circularly ordered boundary nodes $j, k, l, m\in \mathcal{V_B}$, let $f_1=e_{kl}^T\mathcal{L}_{\Gamma}^{\dagger}e_{jm}$ be the potential difference across the boundary nodes $k$ and $l$ when $1\angle0 $ Ampere current is fed into $j^{th}$ node and $m^{th}$ node is grounded and $f_2=e_{jk}^T\mathcal{L}_{\Gamma}^{\dagger}e_{ml}$, defined similarly. Depending on the quadrant in which $f_1$ lies in a complex plane, the following inequalities hold,
	\begin{itemize}
		\item If $\big(\Re\left(f_1\right),\Im\left(f_1\right)\big)\in \mathcal{Q}_1$
		$\implies z^{th}_{j,l}+z^{th}_{k,m} \preceq^{\mathbf{r}} z^{th}_{j,k}+z^{th}_{l,m},$	
		\item If $\big(\Re\left(f_1\right)\Im\left(f_1\right)\big) \in \mathcal{Q}_4$ $\implies$
		\begin{equation*}
			\left(r^{th}_{j,l}+r^{th}_{k,m},\,x^{th}_{j,k}+x^{th}_{l,m} \right)\preceq^{\mathbf{r}} \left(r^{th}_{j,k}+r^{th}_{l,m},\,x^{th}_{j,l}+x^{th}_{k,m}\right),	
		\end{equation*} 
		\item If $\big(\Re\left(f_1\right),\Im\left(f_1\right)\big)\in \mathcal{Q}_2 \implies$
		\begin{equation*}
			\left(r^{th}_{j,k}+r^{th}_{l,m},\,x^{th}_{j,l}+x^{th}_{k,m} \right)\preceq^{\mathbf{r}} \left(r^{th}_{j,l}+r^{th}_{k,m},\,x^{th}_{j,k}+x^{th}_{l,m}\right), 	
		\end{equation*}
		\item If $\big(\Re\left(f_1\right), \Im\left(f_1\right)\big)\in\mathcal{Q}_3$
		$\implies z^{th}_{j,k}+z^{th}_{l,m} \preceq^{\mathbf{r}} z^{th}_{j,l}+z^{th}_{k,m}.$	
	\end{itemize}
	Similarly for $f_2$, following inequalities hold,
	
	\begin{itemize}
		\item If $\big(\Re\left(f_2\right),\Im\left(f_2\right)\big)\in\mathcal{Q}_1$
		$\implies z^{th}_{j,l}+z^{th}_{k,m} \preceq^{\mathbf{r}} z^{th}_{k,l}+z^{th}_{j,m},$	
		\item If $\big(\Re\left(f_2\right),\Im\left(f_2\right)\big)\in\mathcal{Q}_4, \implies$
		\begin{equation*}
			\left(r^{th}_{j,k}+r^{th}_{l,m},\,x^{th}_{k,l}+x^{th}_{j,m} \right)\preceq^{\mathbf{r}} \left(r^{th}_{k,l}+r^{th}_{j,m},\,x^{th}_{j,k}+x^{th}_{l,m}\right),	
		\end{equation*} 
		\item If $\big(\Re\left(f_2\right),\Im\left(f_2\right)\big)\in\mathcal{Q}_2,\implies$
		\begin{equation*}
			\left(r^{th}_{k,l}+r^{th}_{j,m},\,x^{th}_{j,l}+x^{th}_{k,m} \right)\preceq^{\mathbf{r}} \left(r^{th}_{j,k}+r^{th}_{l,m},\,x^{th}_{k,l}+x^{th}_{j,m}\right), 	
		\end{equation*}
		\item If $\big(\Re\left(f_2\right),\Im\left(f_2\right)\big)\in\mathcal{Q}_3$
		$\implies z^{th}_{k,l}+z^{th}_{j,m} \preceq^{\mathbf{r}} z^{th}_{j,l}+z^{th}_{k,m}.$	
	\end{itemize}
\end{thm}
Where, $\mathbf{r}=1$ for RL network, $\mathbf{r}=2$ for RC network.\\
\noindent The Kalmason's inequality takes a special form for networks in $\Gamma_{\beta}$, which is stated in Corrollary \ref{cor:cor8}, and a concise proof is presented in Appendix \ref{App:B1}.
\begin{cor}\label{cor:cor8}
	Consider a circular planar passive electrical network $\Gamma=\left(\mathcal{G},\gamma \right)$, where $\Gamma \in \Gamma_{\beta}$. Choose any four circularly ordered boundary nodes $j, k, l, m \in \mathcal{V_B}$. Based on the sign of term $e_{kl}^T\mathcal{L}_{\mathcal{G}}^{\dagger}e_{jm}$ following inequalities hold,
	\begin{itemize}
		\item if $e_{kl}^T\mathcal{L}_{\mathcal{G}}^{\dagger}e_{jm}>0,$ then the inequality $\mathcal{K}^{kljm}_{>0}: \det \left( {{\mathcal{L_G}}\left[ {j,k} \right]} \right) + \det \left( {{{\mathcal{L_G}}}\left[ {l,m} \right]} \right) \le\det \left( {{{\mathcal{L_G}}}\left[ {j,l} \right]} \right) + \det \left( {{\mathcal{L_G}}\left[ {k,m} \right]} \right)$ holds.
		\item if $e_{kl}^T\mathcal{L}_{\mathcal{G}}^{\dagger}e_{jm}<0,$ then the inequality $\mathcal{K}^{kljm}_{<0}:\det \left( {{\mathcal{L_G}}\left[ {j,l} \right]} \right) + \det \left( {{{\mathcal{L_G}}}\left[ {k,m} \right]} \right) \le\det \left( {{{\mathcal{L_G}}}\left[ {j,k} \right]} \right) +  \det \left( {{\mathcal{L_G}}\left[ {l,m} \right]} \right)$ holds.
	\end{itemize}
	Based on the sign of term $e_{jk}^T\mathcal{L}_{\mathcal{G}}^{\dagger}e_{ml}$ following inequalities hold,
	\begin{itemize}	
		\item if $e_{jk}^T\mathcal{L}_{\mathcal{G}}^{\dagger}e_{ml}<0,$ then the inequality $\mathcal{K}^{jkml}_{<0}:\det \left( {{\mathcal{L_G}}\left[ {j,l} \right]} \right) + \det \left( {{{\mathcal{L_G}}}\left[ {k,m} \right]} \right) \le\det \left( {{{\mathcal{L_G}}}\left[ {k,l} \right]} \right) + \det \left( {{\mathcal{L_G}}\left[ {j,m} \right]} \right)$ holds.
		\item if $e_{jk}^T\mathcal{L}_{\mathcal{G}}^{\dagger}e_{ml}>0,$ then the inequality $\mathcal{K}^{jkml}_{>0}:\det \left( {{\mathcal{L_G}}\left[ {k,l} \right]} \right) + \det \left( {{{\mathcal{L_G}}}\left[ {j,m} \right]} \right) \le\det \left( {{{\mathcal{L_G}}}\left[ {j,l} \right]} \right) +  \det \left( {{\mathcal{L_G}}\left[ {k,m} \right]} \right)$ holds.
	\end{itemize}
\end{cor}
\begin{rem}
	The term $f_1=e_{kl}^T\mathcal{L}_{\mathcal{G}}^{\dagger}e_{jm}>0$ iff there exist two disjoint paths, one from $k$ to $j$ and other from $l$ to $m$. Similarly, term $f_2=e_{jk}^T\mathcal{L}_{\mathcal{G}}^{\dagger}e_{ml}>0$ iff there exist two disjoint paths, one path from $j$ to $m$ and other from $k$ to $l$. The proof concerning this is given in Corrollary-\ref{triangle_gamma_beta} (below equation (\ref{eqn45})) in Appendix \ref{App:B2}. 
\end{rem}
\noindent The Kalmanson's inequality for resistive networks in $\Gamma^R_{\beta}$ takes a special form, which is explained in Corollary-\ref{triangle_gamma_beta}. An alternate approach to derive the Kalmansons inequality for networks in $\Gamma_\beta^R$ is elucidated in Appendix \ref{App:B2}.
\begin{cor}\label{triangle_gamma_beta}
	Given a circular planar passive electrical network $\Gamma=\left(\mathcal{G},\gamma \right)$, where $\Gamma \in \Gamma_{\beta}^{R}$. For any four distinct circularly ordered boundary nodes $j, k, l, m \in \mathcal{V_B}$. Following inequalities, \begin{equation*}
		\begin{split}
			\det\mathcal{L_G}\left[j,k\right] + \det\mathcal{L_G}\left[l,m\right] \le \det\mathcal{L_G}\left[j,l\right] + \det\mathcal{L_G}\left[k,m\right],\\
			\det\mathcal{L_G}\left[j,m\right] + \det\mathcal{L_G}\left[k,l\right] \le \det\mathcal{L_G}\left[j,l\right] + \det\mathcal{L_G}\left[k,m\right], 
		\end{split}	
	\end{equation*}
	holds true.
\end{cor}
\noindent The triangle inequality in Corollary-\ref{cor6_1} and the Kalmanson's inequality in Corollary-\ref{cor:cor8} will serve as constraints in topology reconstruction and set characterization problem.
\section{Problem Formulation}\label{sec:prob_form}
In this work, we aim to reconstruct an unknown network $\Gamma \in \Gamma_{\beta}$ and a set of valid electrical networks satisfying the partially available Thevenin's impedance measurements in $Z^{th}$. Since $\Gamma$ is unknown, the corresponding Laplacian matrix $\mathcal{L}_{\Gamma}$ is also unknown. Therefore, to compute $\mathcal{L}_{\Gamma}$, we relate the $\mathcal{L}_{\Gamma}$ with the known  Thevenin's impedance measurements in set $Z^{th}$ using equation (\ref{imped_dist}). This relation yields a system of multivariate polynomial equations, whose solution yields a reconstructed network and a set of all $\mathcal{L}_{\Gamma}$, i.e., the set of all electrical networks in $\Gamma_{\beta}$ which satisfies the measurements in set $Z^{th}$. Further, we impose the triangle and the Kalmanson's inequalities on these multivariate polynomial equations to get a set of valid electrical networks. To impose the triangle and the Kalmanson's inequalities, consider set $\mathbf{I}_{\Delta}$ and $\mathbf{I}_{\mathcal{K}}$, which are the set of all circularly ordered node indices $\big\{j,k,l\big\}$ and $\big\{j,k,l,m\big\}$ respectively, chosen such that $j,k,l$ and $j,k,l,m$ are circularly ordered. And atleast one node is in $\mathcal{U_B}$ and remaining nodes are from $\mathcal{A}$.

For enforcing the triangle inequality, choose an indices say $\big\{j,k,l\big\} \in \mathbf{I}_{\Delta}$, if only one node, say $k$, in node indices $\big\{j,k,l\big\}$ is in $\mathcal{U_B}$ and $j,l \in \mathcal{A}$ the corresponding triangle inequality is,
\begin{equation*}
	\Re(z^{th}_{j,l}){\mathfrak{T}\left(\mathcal{G}\right)}-\beta\big(\det\mathcal{L_G}\left[j,k\right] + \det\mathcal{L_G}\left[k,l\right]\big) \le 0,	
\end{equation*} 	
here  $z^{th}_{j,l} \in Z^{th}$. Whereas, if more than one node is in $\mathcal{U_{B}}$, then the corresponding inequality is,
\begin{equation*}
	\det\mathcal{L_G}\left[j,l\right]-\left(\det\mathcal{L_G}\left[j,k\right] + \det\mathcal{L_G}\left[k,l\right]\right) \le 0.
\end{equation*} 
Let the set of all such triangle inequalities be labeled as $\mathcal{C}_{\Delta}$. Construction of a set of Kalmanson's inequalities is an involved process and is explained in detail in section \ref{kalman_set_construct}.
\subsection{Construction of a set of Kalmanson's Inequalities}\label{kalman_set_construct}
Let $\mathbf{I}_{\mathcal{K}}=\big\{\{j_q, k_q, l_q, m_q\}:1 \le q \le |\mathbf{I}_{\mathcal{K}}| \big\}$ be the set of index quadruplets. For each quadruplet $\{j_w, k_w, l_w, m_w\} \in \mathbf{I}_{\mathcal{K}}$, corollary \ref{cor:cor8} yields four possible inequalities $\mathcal{K}^{{k_{w}l_{w}j_{w}m_{w}}}_{>0}$, $\mathcal{K}^{{k_{w}l_{w}j_{w}m_{w}}}_{<0}$, $\mathcal{K}^{{j_{w}k_{w}m_{w}l_{w}}}_{>0}$ and $\mathcal{K}^{{j_{w}k_{w}m_{w}l_{w}}}_{<0}$. From these, we construct four composite inequalities: 
\begin{enumerate*}
\item $\mathcal{K}^{w}_1=\mathcal{K}^{{k_{w}l_{w}j_{w}m_{w}}}_{>0} \land \mathcal{K}^{{j_{w}k_{w}m_{w}l_{w}}}_{>0},$
\item $\mathcal{K}^{w}_2=\mathcal{K}^{{k_{w}l_{w}j_{w}m_{w}}}_{>0} \land \mathcal{K}^{{j_{w}k_{w}m_{w}l_{w}}}_{<0},$
\item $\mathcal{K}^{w}_3=\mathcal{K}^{{k_{w}l_{w}j_{w}m_{w}}}_{<0} \land \mathcal{K}^{{j_{w}k_{w}m_{w}l_{w}}}_{>0},$
\item  $\mathcal{K}^{w}_4=\mathcal{K}^{{k_{w}l_{w}j_{w}m_{w}}}_{<0} \land \mathcal{K}^{{j_{w}k_{w}m_{w}l_{w}}}_{<0}$
\end{enumerate*}
Let $\mathcal{K}^{w}=\big\{\mathcal{K}^{w}_1,\mathcal{K}^{w}_2,\mathcal{K}^{w}_3,\mathcal{K}^{w}_4\big\}$ be a set of these composite inequalities for $w^{th}$ index quadruplet. We thus have $|\mathbf{I_{\mathcal{K}}} |$ such sets, one for each quadruplet in set $\mathbf{I}_{\mathcal{K}}$. We denote the collection of all such sets as $\mathcal{K}=\big\{\mathcal{K}^{q}: 1 \le q \le |\mathbf{I}_{\mathcal{K}}|\big\}$. 
From set $\mathcal{K}$, we construct a set $\mathcal{C_{K}}$, containing all possible combinations of composite inequalities drawn from the sets in $\mathcal{K}$, where each combination includes exactly one inequality from each set. Formally, this is expressed as $\mathcal{C_{K}}=\Big\{\big\{x_1 \land x_2 \land \ldots  \land x_{|\mathbf{I}_{\mathcal{K}}|}\big\}:x_1\in \mathcal{K}^{1},x_2\in \mathcal{K}^{2},\ldots,x_{|\mathbf{I}_{\mathcal{K}}|}\in \mathcal{K}^{|\mathbf{I}_{\mathcal{K}}|}\,\mbox{and}\,x_j \ne x_k \forall j \ne k\Big\}$. In set $\mathcal{C_{K}}$, only one of the combination, denoted $\mathbf{c}_{\mathcal{K}}^{\star}\in \mathcal{C_{K}}$ is a true composite inequality which satisfies the unknown network $\Gamma$.

The Kalmanson's inequality associated with a quadruplet $\{j_w, k_w, l_w, m_w\} \in \mathbf{I}_{\mathcal{K}}$ depends on the sign of boundary conditions $e_{k_{w}l_{w}}\mathcal{L}_\mathcal{G}^{\dagger}e_{j_{w}m_{w}}$ and $e_{j_{w}k_{w}}\mathcal{L}_\mathcal{G}^{\dagger}e_{m_{w}l_{w}}$. However, these boundary conditions are not directly computable due to the unavailability of certain boundary nodes within the quadruplet. Consequently, the true combination of inequalities $\mathbf{c}_{\mathcal{K}}^{\star} \in \mathcal{C_{K}}$ cannot be determined. While identifying $\mathbf{c}_{\mathcal{K}}^{\star}$ directly from the Thevenin impedance measurements in $Z^{th}$ is not feasible due to limited availability of measurements, we can construct a subset $\widehat{\mathcal{C}}_{\mathcal{K}} \subseteq {\mathcal{C}}_{\mathcal{K}}$ that is guaranteed to contain $\mathbf{c}_{\mathcal{K}}^{\star}$. This subset, $\widehat{\mathcal{C}}_{\mathcal{K}}=\Big\{\widehat{\mathbf{c}}_{\mathcal{K},i}: \widehat{\mathbf{c}}_{\mathcal{K},i}\, \mbox{satisfies} \, Z^{th} \Big\}$, is constructed by identifying those composite inequalities in $\mathcal{C_K}$ that are consistent with the available measurements. Hence, we aim is to first construct $\widehat{\mathcal{C}}_{\mathcal{K}}$ and use each composite inequality $\widehat{\mathbf{c}}_{\mathcal{K},i} \in \widehat{\mathcal{C}}_{\mathcal{K}}$ in the topology reconstruction algorithm, elucidated in the upcoming sections.

We first formulate the core problem of the topology reconstruction and set enumeration of a class of electrical networks in $\Gamma_{{\beta}}$. We utilize the relation in equation (\ref{imped_dist}) to construct a system of multivariate polynomial equations, i.e., for every Thevenins measurements $z^{th}_{j,k}$ $\mbox{where}\, j,k \in \mathcal{A}$, we have,
\begin{equation}
	z^{th}_{j,k}\det\mathcal{L_G}\left[k,k\right] - \gamma^{-1} \det \mathcal{L_G}\left[j,k\right]=0,
\end{equation}
here, $\Big|\Re\left(z^{th}_{j,k}\right)\Big|=\Big|\Im\left(z^{th}_{j,k}\right)\Big|$. Hence,
\begin{equation}\label{poly_measure}
	\Re\big(z^{th}_{j,k}\big)\det\mathcal{L_G}\left[k,k \right] - \beta \det \mathcal{L_G}\left[j,k\right]=0.
\end{equation}
Let $f_{jk}=\Re\left(z^{th}_{j,k}\right)\det\mathcal{L_G}\left[k,k\right] - \beta \det \mathcal{L_G}\left[j,k\right],\,\forall j,k \in \mathcal{A}$, be multivariate polynomials in elements of Laplacian matrix $\mathcal{L_{G}}$ and $\beta$. Therefore, let $\textbf{w}_l=\left[l_{12},\,\cdots, l_{1n},\,l_{23},\,\cdots \, l_{2n},\,\cdots\,l_{(n-1)n}\right]^\mathbf{T}\in\{0,1\}^{\frac{n\times(n-1)}{2}}$ be a vector of unknown Laplacian matrix   elements, and $\textbf{w}=\left[\textbf{w}_l^\mathbf{T}\,\beta \right]^\mathbf{T}$ be a vector of unknowns. Let a set $\mathcal{F}$ be the collection of all $\frac{|\mathcal{A}|\left(|\mathcal{A}|-1\right)}{2}$ equations, i.e., $\mathcal{F}=\left\{f_{jk}:\forall j,k \in \mathcal{A}\right\}$. From $\mathcal{F}$, we construct a system of polynomial equations $\mathcal{F}\left(\textbf{w}\right)=0$, constrained by the triangle and the Kalmanson's inequality. An additional constraint, i.e., $\mathfrak{T}(\mathcal{G})>0$ is imposed to ensure the connectedness of networks.
\begin{prob}\label{prob:prob1}
	Compute a set of unknown Laplacian matrix $\mathcal{L}_{\Gamma}$ using Thevenin's impedance measurements in set $Z^{th}$. Therefore,   
	\begin{equation*}\label{prob:problem1}	
		\begin{array}{l}
		\begin{array}{l}
				\rm{\textit{Solve\,for}\,\, \textbf{\textbf{w}}}\\
				{\mbox{Subject to}}\\
				\quad \quad{\mathcal{F}}\left( \textbf{w} \right) = 0,
				\quad{\textbf{w}_{l}\odot\left(\textbf{w}_{l}-\mathbf{1}\right)=0},\\ 
				\quad\,\,\,\,\,\, {\mathfrak{T}(\mathcal{G})>0},
				\quad \beta  > 0,
				\quad \mathcal{C}_{\Delta}, \\
			\end{array} \quad \quad \quad \quad \left( \mathcal{P}_i^{\star} \right) \\
			\quad \quad \quad \widehat{\mathbf{c}}_{\mathcal{K},{i}} \in \widehat{\mathcal{C}}_{\mathcal{K}}.
		\end{array} 
	\end{equation*}
	\noindent The problem $\mathcal{P}_i^{\star}$ is solved for every $\widehat{\mathbf{c}}_{\mathcal{K},{i}} \in \widehat{\mathcal{C}}_{\mathcal{K}}$.

	\noindent The solution to such a constrained system of polynomial equations is a set of valid electrical networks satisfying the measurements in set $Z^{th}$.
\end{prob}
\noindent Let the above constrained problem be denoted as $\mathcal{P}_i^{\star}$ and the subproblem of $\mathcal{P}_i^{\star}$ labeled as $\mathcal{P}$ is, 
\begin{equation*}\label{prob:problem1aux}	
	\begin{array}{l}
		\begin{array}{l}
			\rm{\textit{Solve\,for}\,\, \textbf{\textbf{w}}}\\
			{\mbox{Subject to}}\\
			\quad \quad{\mathcal{F}}\left( \textbf{w} \right) = 0,
			\quad{\textbf{w}_{l}\odot\left(\textbf{w}_{l}-\mathbf{1}\right)=0},\\ 
			\quad\,\,\,\,\,\, {\mathfrak{T}(\mathcal{G})>0},
			\quad \beta  > 0,
			\quad \mathcal{C}_{\Delta}.\\
		\end{array} \quad \quad \quad \quad  \left( \mathcal{P} \right) \\
	\end{array}
\end{equation*}
\noindent To construct $\widehat{\mathcal{C}}_{\mathcal{K}}$, an augmented problem $\mathcal{P}_{aug}$, constructed using $\mathcal{P}$, is solved. The feasibility of $\mathcal{P}_{aug}$ helps determine the set of valid Kalmanson's constraint $\widehat{\mathcal{C}}_{\mathcal{K}}$.
Construction of $\widehat{\mathcal{C}}_{\mathcal{K}}$ is explained below in detail. 
\subsection{Construction of $\widehat{\mathcal{C}}_{\mathcal{K}}$: A Combinatorial Approach}
An augmented problem $\mathcal{P}^w_{aug,i}=\mathcal{P} \cup \mathcal{K}^{w}_{i}$ is defined as follows, 
\begin{equation*}\label{prob:problem1}	
	\begin{array}{l}
		\begin{array}{l}
			\rm{\textit{Solve\,for}\,\, \textbf{\textbf{w}}}\\
			{\mbox{Subject to}}\\
			\quad \quad{\mathcal{F}}\left( \textbf{w} \right) = 0,
			\quad{\textbf{w}_{l}\odot\left(\textbf{w}_{l}-\mathbf{1}\right)=0},\\ 
			\quad\,\,\,\,\,\, {\mathfrak{T}(\mathcal{G})>0},
			\quad \beta  > 0,
			\quad \mathcal{C}_{\Delta},\\
		\end{array} \quad \quad \quad  \left( \mathcal{P}^w_{aug,i} \right) \\
		\quad \quad \quad \mathcal{K}^{w}_{i}.
	\end{array}
\end{equation*}
To construct a set of valid Kalmansons constraints $\widehat{\mathcal{C}}_{\mathcal{K}}$, we employ a two-stage process. In the first stage, as detailed in Algorithm \ref{algo:constructC_hat_K1}, we solve $\mathcal{P}^w_{aug,i},\,\forall \, 1\le i \le 4$ and $\forall \,1 \le w \le |\mathbf{I}_{\mathcal{K}}|$, retaining only the consistent composite inequalities $\mathcal{K}^{w}_{i}$ to form the set $\widehat{\mathcal{K}}$.
\renewcommand{\algorithmicrequire}{\textbf{Input:}}
\renewcommand{\algorithmicensure}{\textbf{Output:}}
\begin{algorithm}[h]
	\caption{Stage-1 }\label{algo:constructC_hat_K1}
	\begin{algorithmic}[1]
		\REQUIRE $\mathcal{K}=\Big\{\mathcal{K}^{q}: 1 \le q \le |\mathbf{I}_{\mathcal{K}}|\Big\}$.
		\ENSURE $\widehat{\mathcal{K}}=\Big\{\widehat{\mathcal{K}}^{q}: 1 \le q \le |\mathbf{I}_{\mathcal{K}}|\Big\}$
		\FOR{$q=1\,\mbox{to}\,|\mathbf{I}_{\mathcal{K}}|$}\label{step:iter_q}
		\FOR{$i=1\,\mbox{to}\,4$}\label{step:iter_ii}
		\STATE Solve  $\mathcal{P}^q_{aug,i}=\mathcal{P} \cup \mathcal{K}^{q}_{i}$ \label{step:solve_augg}
		\IF{$\mathcal{P}^q_{aug,i}$ is consistent}
		\STATE $\widehat{\mathcal{K}}^q \leftarrow \mathcal{K}^{q}_{i}$
		\ELSE
		\STATE Go to step \ref{step:iter_ii}
		\ENDIF
		\ENDFOR
		\ENDFOR
	\end{algorithmic}
\end{algorithm}

\noindent In the second stage (Algorithm \ref{algo:constructC_hat_K2}), we first construct an auxiliary set ${\mathcal{C}}^{aux}_{\mathcal{K}}$ from $\widehat{\mathcal{K}}$ in a manner analogous to the construction of $\mathcal{C}_{\mathcal{K}}$ from $\mathcal{K}$. We then formulate and solve auxiliary augmented problems $\mathcal{P}^{aux}_{aug,i}=\mathcal{P} \, \cup \, {\mathbf{c}}^{aux}_{\mathcal{K},i}, \forall\, 1 \le i \le |{\mathcal{C}}^{aux}_{\mathcal{K}}|\, \mbox{and} \, {\mathbf{c}}^{aux}_{\mathcal{K},i} \in {\mathcal{C}}^{aux}_{\mathcal{K}}$. The feasible problems yield valid composite inequalities, which are collected in the final set $\widehat{\mathcal{C}}_{\mathcal{K}}$. This completes the construction of set $\widehat{\mathcal{C}}_{\mathcal{K}}$.
\begin{algorithm}[h]
	\caption{Stage-2 }\label{algo:constructC_hat_K2}
	\begin{algorithmic}[1]
		\REQUIRE ${\mathcal{C}}^{aux}_{\mathcal{K}}$.\vspace{0.08cm}
		\ENSURE $\widehat{\mathcal{C}}_{\mathcal{K}}$		
		\FOR{$i=1\,\mbox{to}\,|{\mathcal{C}}^{aux}_{\mathcal{K}}|$}\label{step:iter_i}
		\STATE Solve  $\mathcal{P}^{aux}_{aug,i}=\mathcal{P} \, \cup \, {\mathbf{c}}^{aux}_{\mathcal{K},i}$ \label{step:solve_aug}
		\IF{$\mathcal{P}^{aux}_{aug,i}$ is consistent}
		\STATE $\widehat{\mathcal{C}}_{\mathcal{K}} \leftarrow {\mathbf{c}}^{aux}_{\mathcal{K},i}$
		\ELSE
		\STATE Go to step \ref{step:iter_i}
		\ENDIF
		\ENDFOR
	\end{algorithmic}
\end{algorithm}

The system of multivariate polynomial equations $\mathcal{F}(\textbf{w})$ are high-degree polynomials, and computing solutions to such a system of equations is computationally demanding. To enable efficient and simple algorithmic solutions, we compute a special set of multivariate basis polynomials ${\mathcal{F}_{\mathbf{G}}}$, known as the \emph{Gr\"obner basis polynomials}, corresponding to $\mathcal{F}$. Specifically, in our problem, we compute the Gr\"obner basis for the extended set $\mathcal{F}_{\mathbf{e}}=\big\{\mathcal{F}\left( \textbf{w} \right),\, \textbf{w}_{l}\odot\left(\textbf{w}_{l}-\mathbf{1}\right)\big\}$, and utilize the resulting Gr\"obner basis in problems $\mathcal{P}^{\star}_{i}$, $\mathcal{P}^{w}_{aug,i}$ and $\mathcal{P}^{aux}_{aug,i}$. A brief overview of Gröbner bases and relevant terminology is presented in the next section. For a comprehensive treatment of Gröbner bases, the reader is referred to \cite{cox1997ideals}, while a concise overview is provided in \cite{buchberger2001grobner}.
\subsection{Gr\"{o}bner Basis} \label{groebner}
Consider a set of polynomials $\mathcal{F} \subset \mathbb{R}\left[ {{l_{12}}, \cdots ,{l_{\left( {n - 1} \right)n},\beta}} \right]$. The ideal of the set of polynomials $\mathcal{F}$, labeled as $\langle \mathcal{F} \rangle$, is defined as;
\begin{defn}\label{defn:ideal}[Ideal of $\mathcal{F}$]
	\begin{equation*}
		\begin{array}{l}
			\left\langle \mathcal{F} \right\rangle  =\\ \left\{ {\sum\limits_{i = 1}^n {{h_i}{f_i}:{h_1}, \ldots, {h_n} \in \mathbb{R}\left[ {{l_{12}}, \cdots ,{l_{\left( {n - 1} \right)n},\beta}} \right]} ,\forall {f_i} \in \mathcal{F}} \right\}.
		\end{array}
	\end{equation*}
\end{defn}
\noindent The solutions satisfying the systems of polynomial equation $\mathcal{F}\left(\textbf{w}\right)=0$ is a set $\mathbf{V}\left(\mathcal{F}\right)$, referred to as the Variety of set $\mathcal{F}$, and is defined as,
\begin{defn}\label{defn:variety}[Variety of set $\mathcal{F}$]
	$\mathbf{V}\left(\mathcal{F}\right)$ is called the affine variety generated by the set of polynomials  $\mathcal{F}$, which is a set,
	\begin{equation}
		\mathbf{V}\left( \mathcal{F} \right) = \left\{ \begin{array}{l}
			\left( {{l_{12}}, \ldots ,{l_{\left( {n - 1} \right)n}},\beta } \right):\\
			{f_i}\left( {{l_{12}}, \ldots ,{l_{\left( {n - 1} \right)n}},\beta } \right) = 0,\forall {f_i} \in \mathcal{F}
		\end{array} \right\}.
	\end{equation}
\end{defn}
\noindent From set $\mathcal{F}$ we compute a set of basis polynomials known as the Gr\"obner basis $\mathcal{F}_{\mathbf{G}}\subset \mathbb{R}\left[ {l_{12}}, \cdots ,{l_{ij}}, \cdots ,l_{n(n-1)},\beta \right]$. The polynomials in ${\mathcal{F}_{\mathbf{G}}}$ allows simple algorithmic solution for computing $\mathcal{F}\left(\textbf{w}\right)=0$. The affine varieties $\mathbf{V}\left( \mathcal{F} \right)$ and $\mathbf{V}\left( \mathcal{F}_{\mathbf{G}} \right)$ are related, which is elucidated below,
\begin{lem}\label{lem:variety1}\cite{cox1997ideals}
	Let $\mathcal{F}$ and ${\mathcal{F}_{\mathbf{G}}}$ be a set of polynomials in $\mathbb{R}\left[ {l_{12}}, \cdots ,{l_{ij}}, \cdots ,l_{n(n-1)},\beta \right]$, and ${\mathcal{F}_{\mathbf{G}}}$ is the Gr\"obner basis of $\mathcal{F}$. Then, $\left\langle\mathcal{F} \right\rangle=\left\langle{\mathcal{F}_{\mathbf{G}}}\right\rangle$.
\end{lem}
\begin{lem}\label{lem:variety2}\cite{cox1997ideals}
	Let $\mathcal{F}$ and $\mathcal{F}_{\mathbf{G}}$ be a set of polynomials in $\mathbb{R}\left[ l_{12}, \cdots, l_{ij}, \cdots, l_{n(n-1)}, \beta \right]$, and $\mathcal{F}_{\mathbf{G}}$ is the Gr\"obner basis of $\mathcal{F}$. Then, from Lemma \ref{lem:variety1}, $\left\langle\mathcal{F} \right\rangle=\left\langle{\mathcal{F}_{\mathbf{G}}}\right\rangle$ is true. Therefore, $\mathbf{V}\left(\mathcal{F}\right) =\mathbf{V}\left( \mathcal{F}_{\mathbf{G}} \right)$.
\end{lem}
\noindent Now consider problem $\mathcal{P}^{\star}_{i}$, and construct an extended set of polynomials $\mathcal{F}_\mathbf{e}=\big\{\mathcal{F}\left( \textbf{w} \right),\, \textbf{w}_{l}\odot\left(\textbf{w}_{l}-\mathbf{1}\right)\big\}$. By computing the Gröbner basis polynomials $\mathcal{F}_{\mathbf{Ge}}$ corresponding to $\mathcal{F}_\mathbf{e}$, we leverage Lemmas \ref{lem:variety1} and \ref{lem:variety2} to establish that $\mathbf{V}\big(\mathcal{F}_{\mathbf{e}}\big) =\mathbf{V}\big( \mathcal{F}_{\mathbf{Ge}}\big)$. This allows us to reformulate problem $\mathcal{P}_{i}^{\star}$ for each $1 \le i \le |\widehat{\mathcal{C}}_{\mathcal{K}}|$ as follows:
\begin{equation*}\label{prob:problem1reform}	
	\begin{array}{l}
		\begin{array}{l}
			\rm{\textit{Solve\,for}\,\, \textbf{\textbf{w}}}\\
			{\mbox{Subject to}}\\
			\quad \quad{\mathcal{F}_{\mathbf{Ge}}\left(\textbf{w}\right)=0},\\ 
			\quad\,\,\,\,\,\, {\mathfrak{T}(\mathcal{G})>0},
			\quad \beta  > 0,
			\quad \mathcal{C}_{\Delta},\\
		\end{array} \quad \quad \quad \quad \quad \quad \quad  \left( \mathcal{P}_{\mathbf{G},i}^{\star} \right) \\
		\quad \quad \quad\, \widehat{\mathbf{c}}_{\mathcal{K},i} \in \widehat{\mathcal{C}}_{\mathcal{K}}.
	\end{array}
\end{equation*} 
Let this reformulation of the problem $\mathcal{P}^{\star}_{i}$ be called $\mathcal{P}^{\star}_{\mathbf{G},i}$, and the subproblem of $\mathcal{P}_{\mathbf{G},i}^{\star}$ labeled as $\mathcal{P}_{\mathbf{G}}$ is, 
\begin{equation*}\label{prob:problem1reform11}	
	\begin{array}{l}
		\begin{array}{l}
			\rm{\textit{Solve\,for}\,\, \textbf{\textbf{w}}}\\
			{\mbox{Subject to}}\\
			\quad \quad{\mathcal{F}_{\mathbf{Ge}}\left(\textbf{w}\right)=0},\\ 
			\quad\,\,\,\,\,\, {\mathfrak{T}(\mathcal{G})>0},
			\quad \beta  > 0,
			\quad \mathcal{C}_{\Delta}.\\
		\end{array}  \quad \quad \quad \quad \quad \quad \quad  \left( \mathcal{P}_{\mathbf{G}} \right) \\
	\end{array}
\end{equation*} 
\noindent The subproblems $\mathcal{P}_{\mathbf{G}}$ and $\mathcal{P}$ are both equivalent since $\mathbf{V}\left(\mathcal{P}\right)=\mathbf{V}\left(\mathcal{P}_{\mathbf{G}}\right)$. Therefore, we can substitute $\mathcal{P}_{\mathbf{G}}$ for $\mathcal{P}$ in both $\mathcal{P}^w_{aug,i}$ and $\mathcal{P}^{aux}_{aug,i}$ to construct $\widehat{\mathcal{C}}_{\mathcal{K}}$.

\noindent The computation of Gr\"obner basis is a standard procedure and is not explained here for brevity; for more details on the computation of Gr\"obner basis, refer to \cite{cox1997ideals}. The Groebner basis computations in this paper were done with \cite{matlab2012matlab}.

\noindent Now, consider a solution $\mathbf{w}$ of $\mathcal{P}_{\mathbf{G},i}^{\star}$ and let $\mathcal{L}_{\Gamma,\mathbf{w}}=\gamma\mathcal{L}_{\mathcal{G}}$ denote the corresponding Laplacian matrix, where $\mathcal{L}_{\mathcal{G}}$ is obtained through a linear transformation $\mathcal{T}$ that maps the vector of Laplacian elements $\mathbf{w}_l$ to a matrix.  
\noindent The problem $\mathcal{P}^{\star}_{\mathbf{G},i}$ may admit multiple solutions, therefore a solution set $\mathcal{N}^i= \Big\{\mathcal{L}^i_{\Gamma,\textbf{w}} \in {\mathbb{R}^{n \times n}}:\textbf{w}\,\mbox{satisfies the problem}\, \mathcal{P}^{\star}_{\mathbf{G},i}\Big\}.$ Each solution set $\mathcal{N}^i \subseteq \mathbf{V}\left(\mathcal{F}_{\mathbf{Ge}}\right)$ and there are $|\widehat{\mathcal{C}}_{\mathcal{K}}|$ such solution sets.  \textit{The problem then becomes to compute all the sets $\mathcal{N}^1\,,\mathcal{N}^2,\,\mathcal{N}^3,\, \ldots ,\, \mathcal{N}^{|\widehat{\mathcal{C}}_{\mathcal{K}}|}$, which entails solving problems $\mathcal{P}^{\star}_{\mathbf{G},1}\,, \mathcal{P}^{\star}_{\mathbf{G},2}\,, \mathcal{P}^{\star}_{\mathbf{G},2}\,,\ldots,\,\mathcal{P}^{\star}_{\mathbf{G},|\widehat{\mathcal{C}}_{\mathcal{K}}|}$}.

\noindent To check for the consistency and the number of solutions to a system of polynomial equation ${\mathcal{F}_{\mathbf{Ge}}\left(\textbf{w}\right)=0}$, we define two terms, i.e., the Initial ideal and the standard monomials. A detailed discussion regarding this is done in section-\ref{sec:discuss}.
\begin{defn}[Leading Monomial]
	A leading monomial of a polynomial say $f \in \mathbb{R}\left[ {{l_{12}}, \cdots ,{l_{\left( {n - 1} \right)n},\beta}} \right]$, labeled as $LM\left(f\right)$, is a monomial with highest degree.	
\end{defn}
\begin{defn}[Initial Ideal of $\left\langle \mathcal{F} \right\rangle$] The initial ideal of a graded $\langle \mathcal{F} \rangle$ denoted as $in_{\prec} \big( \left\langle \mathcal{F} \right\rangle \big)$ is the ideal generated by leading monomials of polynomials in $\mathcal{F}$,  
	\begin{equation}
		in_{\prec} \big( \left\langle \mathcal{F} \right\rangle \big) = \left\langle LM\left(f\right) | f \in \left\langle \mathcal{F} \right\rangle \right\rangle.
	\end{equation}
\end{defn}
\noindent Graded ideal here refers to the defined monomial ordering. Finally, the standard monomial corresponding to the $in_{\prec} \left( \left\langle \mathcal{F} \right\rangle \right)$ is,
\begin{defn}[Standard monomial] A monomial say $l^{\alpha}$, is said to be a standard monomial if $l^{\alpha} \notin in_{\prec} \left( \left\langle \mathcal{F} \right\rangle \right)$. Here, $l^{\alpha}$ represents a monomial $l^{\alpha_1}_{12}l^{\alpha_2}_{13}\cdots l^{\alpha_m}_{(n-1)n}$, where $m=\frac{n(n-1)}{2}$ and each $\alpha_{i} \in \mathbb{Z}^{+}\, , \forall i \in \{1,2,\ldots,\frac{n(n-1)}{2}\}$. 	
\end{defn}
\noindent In the following section, we provide a concise overview of the topology reconstruction algorithm. Subsequently, Section \ref{sec:num_exmp} presents an illustrative example to demonstrate the application and effectiveness of the developed method. 
\section{Topology Reconstruction Algorithm}\label{sec:topoalgo}
A comprehensive flowchart illustrating the topology reconstruction algorithm for constructing $\mathcal{N}^{i},\, \forall 1 \le i \le |\widehat{\mathcal{C}}_{\mathcal{K}}|$ is presented in Figure \ref{fig:algoflow}.
\begin{figure}[h]
	\begin{tikzpicture}[node distance=2cm]
		\centering
		\node (start) [startstop, text width=4.5cm] {\large Construct $Z^{th}$ by performing experiments on unknown $\Gamma \in \Gamma^\beta$};
		
		\node (setF) [startstop, below of=start, text width=3cm] { \large Construct set $\mathcal{F}_{\mathbf{e}}$ for measurements in $Z^{th}$};
		
		\node (setFGe) [startstop, right of=setF, xshift=2cm, text width=3cm] {\large Compute the \\ Gr\"obner basis polynomials $\mathcal{F}_{\mathbf{Ge}}$};
		
		\node (setconstraint) [startstop, below of=setF, text width=2.5cm, xshift=-0.5cm] {\large Construct the set $\mathcal{C}_\Delta$};		
		
		\node (solve) [startstop, below of=setFGe, xshift=-1.4cm,text width=1cm] {\large Solve \\ $\mathcal{P}_{\mathcal{G},i}^{\star}$};
		
		\node (construct_C_hat_K) [startstop, right of=solve, xshift=0.8cm, text width=2cm] {\large Construct set  $\widehat{\mathcal{C}}_{\mathcal{K}}$};
		
		\node (N) [startstop, below of=solve] {\large Construct a set of all electrical network $\mathcal{N}^i$.};
		
		\draw [arrow] (start) -- (setF);
		\draw [arrow] (setF) -- (setFGe);
		\draw [arrow] (setFGe) -- (solve);
		\draw [arrow] (setconstraint) -- (solve);
		\draw [arrow] (construct_C_hat_K) -- (solve);
		\draw [arrow] (solve) -- (N);
	\end{tikzpicture}
	\caption{A flow of Topology reconstruction algorithm for constructing set $\mathcal{N}^{i}$.}
	\label{fig:algoflow}
\end{figure}
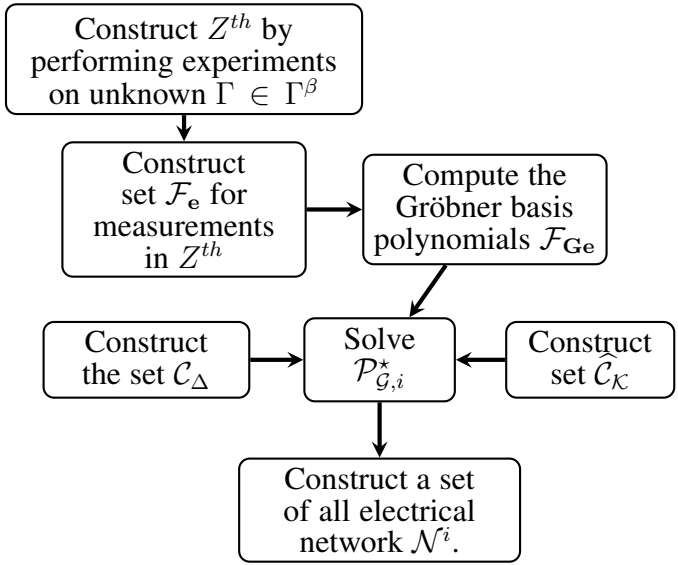
\section{Numerical Example}\label{sec:num_exmp}
To illustrate the developed algorithm, an unknown 5-node network $\Gamma$ is considered, with nodes 2 and 3 are inaccessible for boundary measurements, thus $\mathcal{U}=\{2,3\}$. The index sets are $\mathbf{I}_{\mathcal{K}}=\big\{\left\{1,2,3,4\right\},\,\left\{2,3,4,5\right\},\,\left\{3,4,5,1\right\},\,\left\{4,5,1,2\right\},\,\left\{5,1,2,3\right\}\big\}$ and $\mathbf{I}_{\Delta}=\big\{\left\{1,2,3\right\},\,\left\{2,3,4\right\},\,\left\{3,4,5\right\},\,\left\{5,1,2\right\},\,\left\{1,2,4\right\},$\\$\,\left\{1,3,4\right\},\,\left\{1,3,5\right\},\,\left\{2,3,5\right\},\,\left\{2,4,5\right\}\big\}$. The available Thevenin impedance measurement are $Z^{th}\left(\Gamma\right)=\big\{z^{th}_{1,4}=2.2857+i2.2857, z^{th}_{1,5}=$ $1.2381+i1.2381, z^{th}_{4,5}=1.2381+i1.2381\big\}$. The measurements indicate that the network $\Gamma$ is a RL network, i.e. $\Gamma \in \Gamma_{\beta}^{RL}$.  The Laplacian matrix $\mathcal{L}_{\Gamma}=\gamma \mathcal{L}_{\mathcal{G}} \in \mathbb{R}^{5 \times 5}$ is unknown, and hence, a vector of unknown elements of the Laplacian matrix $\mathcal{L}_{\mathcal{G}}$ is $\mathbf{w}_l=\left[l_{12}, l_{13},\ldots, l_{15}, l_{23},\ldots,l_{25},\ldots,l_{45} \right]^\mathbf{T}\in \{0,1\}^{10 \times 1}$, and $\mathbf{w}=\left[\mathbf{w}^{\mathbf{T}}_l\,\, \beta \right]^{\mathbf{T}}$ is a vector of unknowns. We construct polynomials $f_{14}=2.2857\det \mathcal{L}_{\mathcal{G}}\left[1,1\right]-\beta \det \mathcal{L}_{\mathcal{G}}\left[1,4\right]$, $f_{15}=1.2381\det \mathcal{L}_{\mathcal{G}}\left[1,1\right]-\beta \det \mathcal{L}_{\mathcal{G}}\left[1,5\right]$, $f_{45}=1.2381\det \mathcal{L}_{\mathcal{G}}\left[4,4\right]-\beta \det \mathcal{L}_{\mathcal{G}}\left[4,5\right]$, corresponding to each Thevenin's impedance measurement $z^{th}_{1,4}$, $z^{th}_{1,5}$ and $z^{th}_{4,5}$ respectively. The resulting set of multivariate polynomials $\mathcal{F}=\big\{f_{14},\,f_{15},\,f_{45} \big\}$ and the extended set of polynomials is $\mathcal{F}_\mathbf{e}=\big\{f_{14},\,f_{15},\,f_{45},\,\textbf{w}_{l}\odot\left(\textbf{w}_{l}-\mathbf{1}\right) \big\}$. Compute the Gr\"{o}bner basis polynomials $\mathcal{F}_\mathbf{Ge}$ for $\mathcal{F}_\mathbf{e}$ and construct the subproblem $\mathcal{P}_\mathbf{G}$, as described in section \ref{groebner}. From the index set $\mathbf{I}_{\mathcal{K}}$, we generate a set $\mathcal{K}=\big\{ \mathcal{K}^q:1 \le q \le 5 \big\}$ of composite Kalmanson's inequalities. Notably, $\mathcal{K}^1=\Big\{\mathcal{K}^{{2314}}_{>0} \land \mathcal{K}^{{1243}}_{>0},\,\mathcal{K}^{{2314}}_{>0} \land \mathcal{K}^{{1243}}_{<0},\,\mathcal{K}^{{2314}}_{<0} \land \mathcal{K}^{{1243}}_{>0},\,\mathcal{K}^{{2314}}_{<0} \land \mathcal{K}^{{1243}}_{<0} \Big\}$ corresponding to the quadruplet $\big\{1,2,3,4\big\} \in \mathbf{I}_{\mathcal{K}}$, with similar remaining sets $\mathcal{K}^2$, $\mathcal{K}^3$, $\mathcal{K}^4$, $\mathcal{K}^5$. Then, using set $\mathcal{K}$, we build a set $\mathcal{C_K}$ which has  $4^5$ composite inequalities, among which only one inequality, i.e., $\mathbf{c}^\star_{\mathcal{K}} \in \mathcal{C_K}$ is true. However, due to inaccessible boundary nodes, there may be multiple composite inequalities that satisfy the limitedly available boundary measurements $Z^{th}$, we therefore compute a subset $\widehat{\mathcal{C}}_{\mathcal{K}} \subseteq \mathcal{C}_{\mathcal{K}}$ that is guaranteed to contain $\mathbf{c}^\star_{\mathcal{K}}$. To determine the valid Kalmanson constraints $\widehat{\mathcal{C}}_{\mathcal{K}} \subseteq \mathcal{C}_{\mathcal{K}}$, we employ a two-stage process. \\
\textbf{Stage-1:} We iteratively solve augmented problem $\mathcal{P}^w_{aug,i}=\mathcal{P} \cup \mathcal{K}^w_i,\, \forall 1 \le i \le 4$ and $\forall 1 \le w \le 5$, as detailed in Algorithm \ref{algo:constructC_hat_K1}. This process identifies feasible composite inequalities $\mathcal{K}^w_i$ for each index quadruplet, which are then collected to form the set $\widehat{\mathcal{K}}$. In the current  example, $\widehat{\mathcal{K}}=\big\{\widehat{\mathcal{K}}^1, \widehat{\mathcal{K}}^2, \widehat{\mathcal{K}}^3, \widehat{\mathcal{K}}^4, \widehat{\mathcal{K}}^5\big\}$ is found to be,  
\begin{enumerate*}
	\item $\widehat{\mathcal{K}}^1=\big\{\mathcal{K}^{{2314}}_{>0} \land \mathcal{K}^{{1243}}_{>0},\,\mathcal{K}^{{2314}}_{<0} \land \mathcal{K}^{{1243}}_{<0}\big\}$,
	
	\item $\widehat{\mathcal{K}}^2=\big\{\mathcal{K}^{{3425}}_{>0} \land \mathcal{K}^{{2354}}_{>0},\,\mathcal{K}^{{2314}}_{<0} \land \mathcal{K}^{{1243}}_{<0}\big\}$,
	
	\item $\widehat{\mathcal{K}}^3=\big\{\mathcal{K}^{{4531}}_{>0} \land \mathcal{K}^{{3415}}_{>0}\big\}$,
	
	\item   $\widehat{\mathcal{K}}^4=\big\{\mathcal{K}^{{5142}}_{>0} \land \mathcal{K}^{{4521}}_{>0}\big\}$,
	
	\item   $\widehat{\mathcal{K}}^5=\big\{\mathcal{K}^{{1253}}_{>0} \land \mathcal{K}^{{5132}}_{>0},\,\mathcal{K}^{{1253}}_{<0} \land \mathcal{K}^{{5132}}_{<0}\big\}$.
\end{enumerate*}

\noindent \textbf{Stage-2:} We first construct an auxiliary set ${\mathcal{C}}^{aux}_{\mathcal{K}}$ from $\widehat{\mathcal{K}}$, which has 8 composite inequalities. Then, we solve auxiliary augmented problem $\mathcal{P}^{aux}_{aug,i}=\mathcal{P} \cup {\mathbf{c}}^{aux}_{\mathcal{K},i},\, \forall 1 \le i \le 8$, as described in Algorithm-\ref{algo:constructC_hat_K2}. The feasible solutions of these problems yield valid composite inequalities, which constitute the final set $\widehat{\mathcal{C}}_{\mathcal{K}} \subseteq \mathcal{C}_{\mathcal{K}}$. In this example, $\widehat{\mathcal{C}}_{\mathcal{K}}$ contain only one composite inequality namely: $\widehat{\mathcal{C}}_{\mathcal{K}}=\Big\{\mathcal{K}^{{2314}}_{>0} \land \mathcal{K}^{{1243}}_{>0}\land\mathcal{K}^{{3425}}_{>0} \land \mathcal{K}^{{2354}}_{>0}\land \mathcal{K}^{{4531}}_{>0} \land \mathcal{K}^{{3415}}_{>0}\land\mathcal{K}^{{5142}}_{>0} \land \mathcal{K}^{{4521}}_{>0}\Big\}$. Therefore, this single composite inequality is $\mathbf{c}_{\mathcal{K}}^{\star}=\widehat{\mathcal{C}}_{\mathcal{K}}$. The solution to problem $\mathcal{P}^{\star}_{\mathbf{G}}$ associated to this constraint $\widehat{\mathcal{C}}_{\mathcal{K}}$ is an electrical network $\widehat{\Gamma}$, shown in Figure \ref{fig:fig5}.
\begin{figure}[h]
	\centering
	\includegraphics*[scale=0.65]{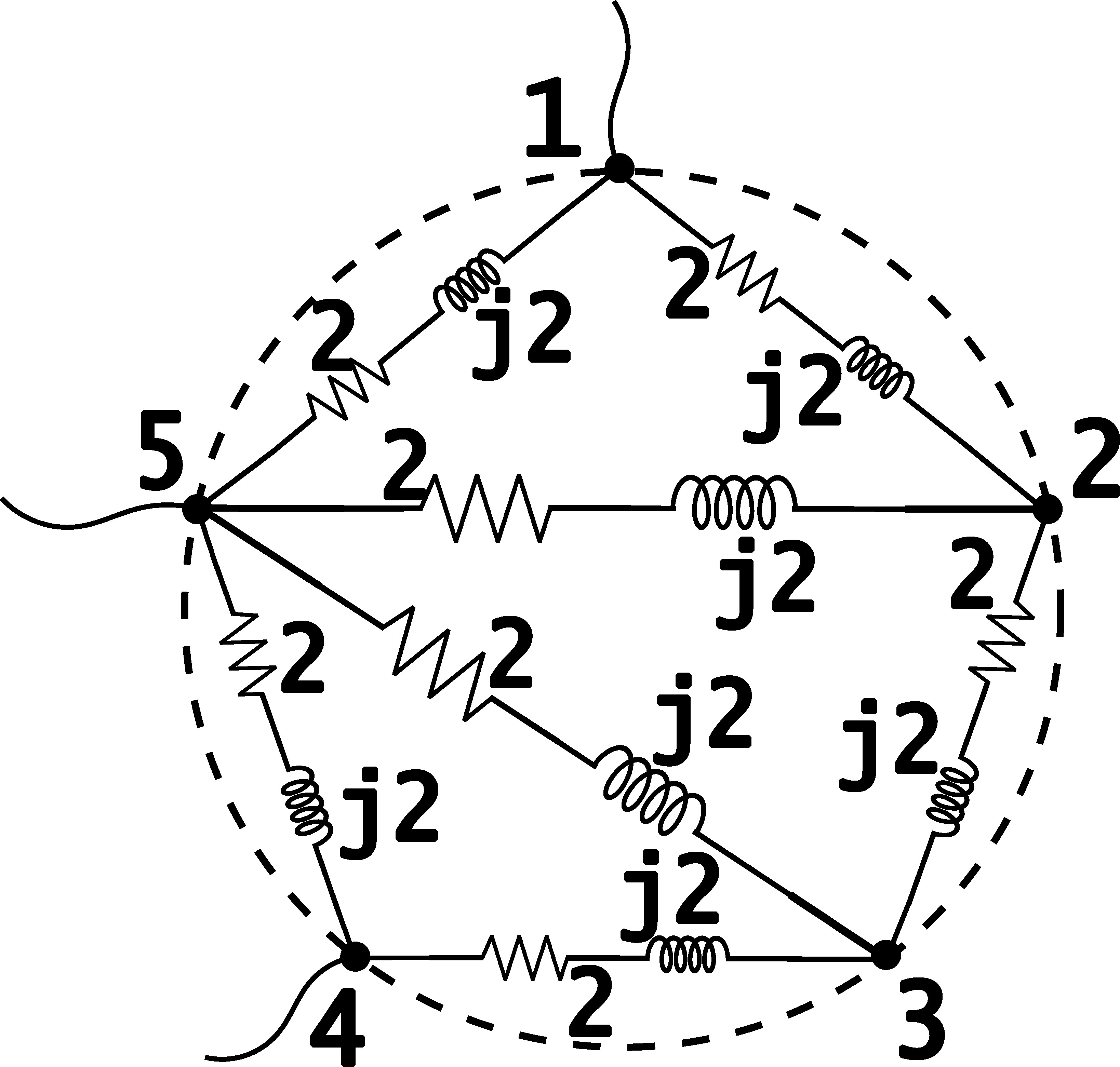}
	\caption{A valid electrical network $\widehat{\Gamma}$ in set $\mathcal{N}^1$}
	\label{fig:fig5}
\end{figure}
\noindent The set $\mathcal{N}^1$ contains only this reconstructed network $\widehat{\Gamma}$. Note that, without additional constraints, there are $|\mathbf{V}\big( \mathcal{F}_{\mathbf{Ge}}\big)|=148$ potential electrical network which satisfy the boundary measurements $Z^{th}$.

Now consider the same example, with $\mathcal{U}=\left\{2,3,4\right\}$. Then $Z^{th}=\left\{z^{th}_{1,5}=1.2381+i1.2381\right\}$. In this case, $\widehat{\mathcal{C}}_{\mathcal{K}} = \mathcal{C}_{\mathcal{K}}$, and solving problem $\mathcal{P}_{\mathbf{G},i}^{\star},\, \forall 1 \le i \le 4^5$ yields $4^5$ sets of valid electrical networks. Figure \ref{fig:figex} illustrates 18 valid networks from one solution set $\mathcal{N}^1$ corresponding to a specific composite inequality $\widehat{\mathcal{C}}_{\mathcal{K},1}=\Big\{\mathcal{K}^{{2314}}_{>0} \land \mathcal{K}^{{1243}}_{>0}\land\mathcal{K}^{{3425}}_{>0} \land \mathcal{K}^{{2354}}_{>0}\land \mathcal{K}^{{4531}}_{>0} \land \mathcal{K}^{{3415}}_{>0}\land\mathcal{K}^{{5142}}_{>0} \land \mathcal{K}^{{4521}}_{>0}\Big\} \in \widehat{\mathcal{C}}_{\mathcal{K}}$.  Among these, the boxed network represents the reconstructed network, identical to the original network $\Gamma$.
\begin{figure}[h]
	\centering
	\includegraphics*[scale=0.6]{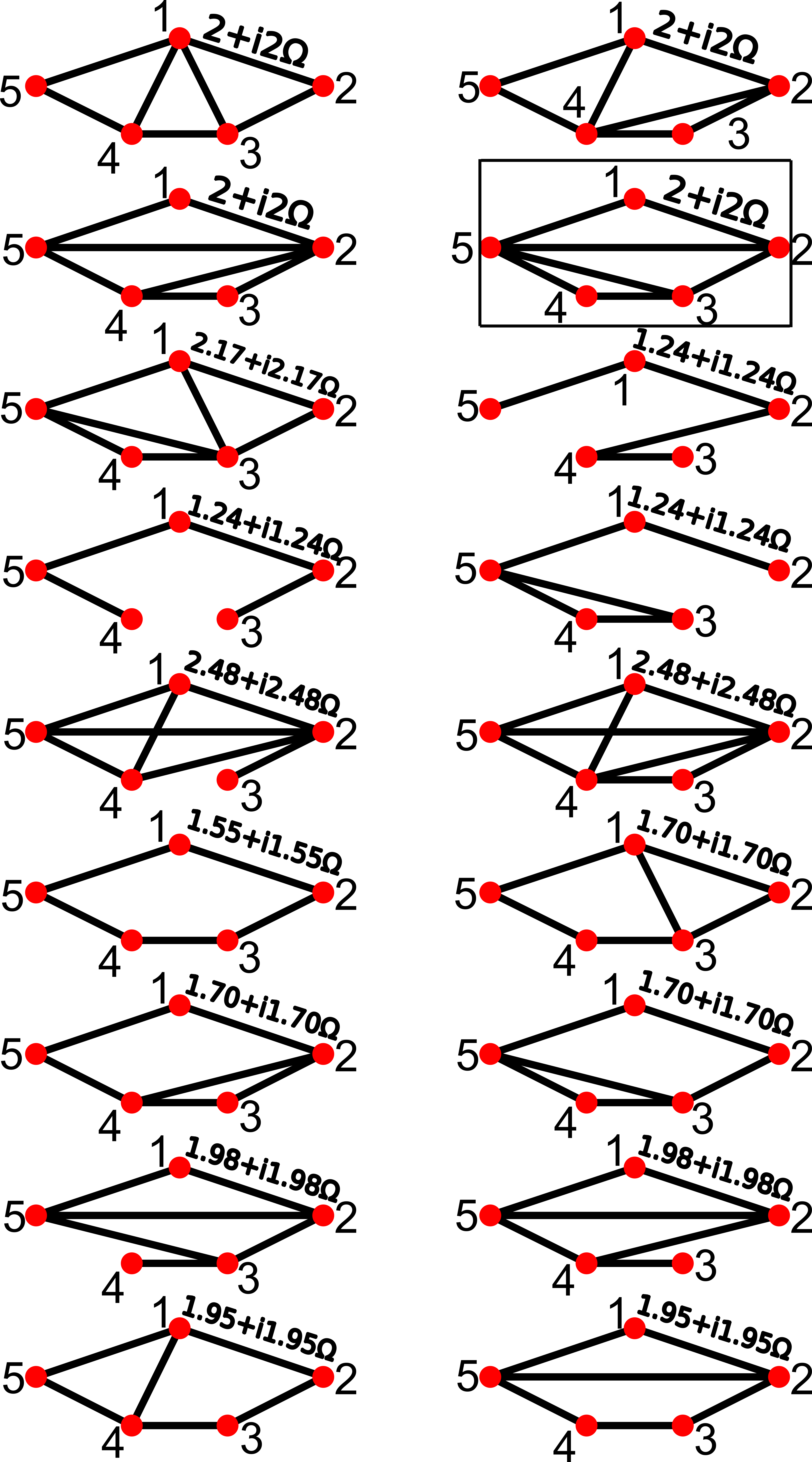}
	\caption{A valid electrical network in set $\mathcal{N}^1$}
	\label{fig:figex}
\end{figure}
\section{Discussion}\label{sec:discuss}
\subsection{Number of Reconstructed Electrical Network}
Let the number of boundary nodes be $|\mathcal{A}|=n_{A}$. Then, the number of multivariate polynomial equations $|\mathcal{F}|$ are $\frac{n_{A}\left(n_{A}-1\right)}{2}$, and the number of unknowns, i.e., the elements of the unknown Laplacian matrix $\mathcal{L_{G}}$, are $\frac{n(n-1)}{2}$,   
\begin{itemize}
	\item The system of polynomial equations $\mathcal{F}(\textbf{w})=0$ always admits trivial solutions with $\mathbf{w}_l=\textbf{0}$.
	\item If $n_A=n$ and one of the variable in $\textbf{w}$ (assume $\beta$) is known, then the number of polynomial equations is equal to number of unknowns. According to B\'ezout's theorem \cite{garcia1980number}, the number of solutions to such $n$ system of polynomial equations in $n$ variable are bounded by $\prod\limits_{i = 1}^n {{d_i}}$, where $d_i$ is the degree of $i^{th}$ polynomial. A better estimate on the bound of the number of solutions to $n$ system of polynomial equations in $n$ variable is given by Bernstein-Kushnirenko-Khovanskii (BKK) theorem\cite{bernshtein1975number}, which states that the number of isolated solutions are atmost the mixed volume of newton polytope corresponding to $n$ polynomial equations. Computation of mixed volume for a high degree multivariate polynomial with many monomials is a difficult task, and is an active area of research \cite{lee2011mixed}. 
	\item If $n_A < n$, we have under-determined system of polynomial equation. The B\'ezouts and BKK theorem does not give conclusive results in such cases. To study the number of solutions we examine the number of standard monomials of  $\langle\mathcal{F}\rangle$. The solution set $\mathbf{V}\left(\mathcal{F}\right)$ is finite iff the set of standard monomials is finite, in which case $|\mathbf{V}\left(\mathcal{F}\right)|$ is equal to number of standard monomials. In this scenario, the solutions in $\mathbf{V}\left(\mathcal{F}_\mathbf{e}\right)$ that satisfy the given constraints constitute the set of valid electrical networks.
\end{itemize}
\subsection{Computational Complexity}
The worst case complexity of computing Gr\"obner basis is doubly exponential with respect to number of variables \cite{bardet2015complexity}. The computational complexity is influenced by two key factors, i.e., 1) degree of polynomial and 2) term ordering.
\begin{enumerate}
	\item The higher the polynomial degree, the higher the computational complexity. 
	\item Term ordering is a way to rank the monomials of the given polynomial. Different choices of term orderings can lead to different Gr\"obner bases; these bases have different numbers of monomials. Thus, the term order directly impacts the computation of $\mathcal{S}$ polynomial in the Buchberger algorithm \cite{cox1994ideals}. There is no standard method of determining an optimal term ordering for a given problem.
\end{enumerate}

\noindent To reduce the computational complexity, one can incorporate apriori knowledge about the network structure. Given our focus on the topology reconstruction of planar networks, we exploit the fact that a maximal planar graph with $n_b$ boundary nodes has $\left(3n_b-6\right)$ edges. For example, with $n_b=5$, a maximal planar network has $9$ edges, 
let this maximal planar network be $\Gamma_5^{max}=\left(\mathcal{G}^{max}_5,\gamma\right)$. Here, $\mathcal{G}^{max}_5$ is an underlying maximal planar graph $\left(\mbox{known to us}\right)$. Let $\Gamma_5=\left(\mathcal{G}_5,\gamma\right)$ be an unknown network which is to be reconstructed. We assume that the unknown graph $\mathcal{G}_5 \subset \mathcal{G}_5^{max}$, since multiple maximal planar network on given $n_{b}$ boundary nodes are possible. The number of unknowns in its Laplacian matrix $\mathcal{L}(\mathcal{G}_5)$ is reduced from 10 to 9, thereby reducing the degree and size of polynomials, constructed in above examples. Incorporating a maximum degree constraint on nodes within $\mathcal{P}_{\mathbf{G},i}^{\star}$ yields a more refined set of solutions.
\subsection{Constructing set $\widehat{\mathcal{C}}_{\mathcal{K}}$}
The developed algorithm for constructing a set $\widehat{\mathcal{C}}_{\mathcal{K}}$ serves as a method to explore the internal structure of an unknown electrical network, since elements of set $\widehat{\mathcal{C}}_{\mathcal{K}}$ in general depends on sign of the boundary conditions  $f_1=e_{kl}^T\mathcal{L}_{\mathcal{G}}^{\dagger}e_{jm}$, $f_2=e_{jk}^T\mathcal{L}_{\mathcal{G}}^{\dagger}e_{ml}$, where the condition $f_1>0$ implies the existence of disjoint paths from $k$ to $j$ and $l$ to $m$, and $f_2>0$ indicates disjoint paths from $j$ to $m$ and other from $k$ to $l$. 
However, Constructing $\widehat{\mathcal{C}}_{\mathcal{K}}$ involves repeatedly formulating augmented problems and checking whether a solution exists or not, making the process computationally resource intensive and time consuming. Therefore a more efficient approach to construct set $\widehat{\mathcal{C}}_{\mathcal{K}}$ is a subject of future research.

To reduce the computational burden associated with computing the set $\widehat{\mathcal{C}}_{\mathcal{K}}$, practitioner's can utilize apriori knowledge of disjoint paths to construct the set $\widehat{\mathcal{C}}_{\mathcal{K}}$ directly, avoiding the need to solve the augmented problems. If, only partial information about the disjoint paths is available, it can be incorporated into the intermediate set $\widehat{\mathcal{K}}$ and then construct set $\widehat{\mathcal{C}}_{\mathcal{K}}$. Conversely, if practitioners aims to build a network with specific disjoint paths, this requirement can be conveniently incorporated into the algorithm thereby reducing the computational complexity. 
\section{Conclusion}
In this paper, we have presented a novel algorithm for topology reconstruction and characterization of a set of all electrical networks that are consistent with the limitedly available Thevenin's impedance measurements computed experimentally from an unknown circular planar passive electrical network (R, RL, or RC). The unknown network is assumed to belong to a class of circular planar passive networks with equal edge impedances having equal magnitude real and imaginary parts. The developed algorithm leverages a derived relationship between Thevenin impedances and the unknown Laplacian matrix to formulate a system of multivariate polynomial equations. To find solutions efficiently, Gröbner basis polynomials are employed to solve the system, subject to the triangle and Kalmanson's inequalities, along with a graph connectedness constraint.

A key challenge in applying Kalmanson's inequalities is the unavailability of certain boundary terminals, which prevents direct computation of necessary boundary conditions. To overcome this, we introduce a method to construct a set of valid Kalmanson's inequalities $\widehat{\mathcal{C}}_{\mathcal{K}}$, by solving augmented problems $\mathcal{P}^{w}_{aug,i}$ and auxiliary problems $\mathcal{P}^{aux}_{aug,i}$.

Furthermore, in this article we prove that the Kalmanson's and the triangle inequalities holds for general circular planar passive R, RL, and RC networks, provided the Thevenin's impedances lie within a convex cone and specific boundary conditions are met. This result establishes the theoretical foundation for our algorithm and extends the applicability of these inequalities.

While the developed algorithm demonstrates promising results for the specific class of networks considered, future work will focus on extending its applicability to general R, RL, and RC networks and investigating a computationally efficient method for identifying set $\widehat{\mathcal{C}}_{\mathcal{K}}$ is needed. Additionally, we aim to incorporate the handling of noisy measurements into the algorithm to enhance its practicality. 
\appendix
\subsection{Triangle Inequality}\label{App:A}
\begin{center}
	\textcolor{black}{\textbf{Proof for Theorem-\ref{thm:tri_ineq_RL_RC1}}}\\
\end{center}
\begin{proof}
	Consider three boundary nodes $j$, $k$ and $l$ such that $1\le j < k < l \le n_b$. The  Thevenin's impedance $z^{th}_{j,l}$ is defined as,
	\begin{equation}\label{eqn:tri1}
		z^{th}_{j,l}=e_{jl}^T \mathcal{L}_{\Gamma}^{\dagger}e_{jl}=\left(e_{jk}+e_{kl}\right)^T \mathcal{L}_{\Gamma}^{\dagger} \left(e_{jk}+e_{kl}\right).
	\end{equation}    
	
	Therefore, we have, 
	\begin{equation}\label{eqn:tri3}
		\begin{array}{ll}
			z^{th}_{j,l}&=e_{jk}^T \mathcal{L}_{\Gamma}^{\dagger}e_{jk} + e_{kl}^T \mathcal{L}_{\Gamma}^{\dagger}e_{kl} + 2e_{jk}^T \mathcal{L}_{\Gamma}^{\dagger}e_{kl},\\
			&=z^{th}_{j,k} + z^{th}_{k,l} + 2e_{jk}^T \mathcal{L}_{\Gamma}^{\dagger}e_{kl}=z^{th}_{j,k} + z^{th}_{k,l} + t,\\
			\left(r^{th}_{j,l},x^{th}_{j,l}\right)&= \left(r^{th}_{j,k}+r^{th}_{k,l}+\Re(t),\,x^{th}_{j,k}+x^{th}_{k,l}+\Im(t)\right).
		\end{array}
	\end{equation}
	Here $t=e_{jk}^T \mathcal{L}_{\Gamma}^{\dagger}e_{kl}$, is the potential difference across boundary nodes $j$ and $k$ when $1\angle0 $ Ampere current is fed into $k^{th}$ node and $l^{th}$ node is grounded.

	If,
	\begin{itemize}
		\item  $\Re \left( t \right) \ge 0,\,\Im \left( t \right) \ge 0,\mbox{then}\, \left(r^{th}_{j,k}+r^{th}_{k,l}, x^{th}_{j,k}+x^{th}_{k,l}\right)\preceq^\mathbf{r} \left(r^{th}_{j,l},x^{th}_{j,l}\right) \implies z^{th}_{j,k}+z^{th}_{k,l}\preceq^{\mathbf{r}}z^{th}_{j,l},$\vspace{0.1cm}
		\item $\Re \left( t \right) \ge 0,\,\Im \left( t \right) \le 0,\mbox{then} \left(r^{th}_{j,k}+r^{th}_{k,l}, x^{th}_{j,l}\right)\preceq^\mathbf{r} \left(r^{th}_{j,l},x^{th}_{j,k}+x^{th}_{k,l}\right).$\vspace{0.1cm}
		\item $\Re \left( t \right) \le 0,\,\Im \left( t \right) \ge 0,\mbox{then}\,  \left(r^{th}_{j,l}, x^{th}_{j,k}+x^{th}_{k,l}\right)\preceq^\mathbf{r} \left(r^{th}_{j,k}+r^{th}_{k,l},x^{th}_{j,l}\right).$\vspace{0.1cm}
		\item $\Re \left( t \right) \le 0,\,\Im \left( t \right) \le 0,\mbox{then}\,\left(r^{th}_{j,l},x^{th}_{j,l}\right)\preceq^\mathbf{r} \left(r^{th}_{j,k}+r^{th}_{k,l}, x^{th}_{j,k}+x^{th}_{k,l}\right)\implies z^{th}_{j,l} \preceq^{\mathbf{r}} z^{th}_{j,k}+z^{th}_{k,l}.$
	\end{itemize}
	Here $\mathbf{r}=1$ for any $\Gamma \in \Upsilon^{RL}$ and $\mathbf{r}=2$ for any $\Gamma \in \Upsilon^{RC}$. 
\end{proof}

\subsection{Triangle Inequality on $\Gamma_{\beta}$}\label{App:A1}
\begin{center}
	\textcolor{black}{\textbf{Proof for Corollary-\ref{cor6_1}}}\\
\end{center}
\begin{proof}
	We know that,
	\begin{equation}\label{eqn:tri3}
		z^{th}_{j,l}=z^{th}_{j,k} + z^{th}_{k,l} + 2e_{jk}^T \mathcal{L}_{\Gamma}^{\dagger}e_{kl}.
	\end{equation}
	Let $t=e_{jk}^T \mathcal{L}_{\Gamma}^{\dagger}e_{kl}$. Since $\Gamma \in \Gamma_{\beta}$, we write, $t = \gamma^{-1} e_{jk}^T \mathcal{L}_{\mathcal{G}}^{\dagger}e_{kl}$. The term $e_{jk}^T \mathcal{L}_{\mathcal{G}}^{\dagger}e_{kl}$ is the potential difference across boundary nodes $j$ and $k$ when $1\angle0$ Ampere current is fed into $k^{th}$ node and $l^{th}$ node is grounded, therefore $v_{jl}<v_{kl}$. Hence,
	$e_{jk}^T \mathcal{L}_{\mathcal{G}}^{\dagger}e_{kl} <0$. Let $\left(\Re(t),\Im(t)\right)$ be the ordered pair corresponding to $t$, here $|\Re(t)|=|\Im(t)|$. Pair $\left(\Re(t),\Im(t)\right) \in \mathcal{Q}_3$ for networks in $\Gamma_{\beta}^{RL}$ therefore
	\begin{equation}
		\Re(z^{th}_{j,l}) \le \Re(z^{th}_{j,k}) + \Re(z^{th}_{k,l}), 
	\end{equation}	
	is true. 
	Using a result in equation (\ref{imped_dist}), in Theorem \ref{impedance_dist} we get,
	\begin{equation}\label{eqn:corro1}
		\beta \frac{{\det  {\mathcal{L_G}\left[ {j,l} \right]} }}{{\det  {\mathcal{L_G}\left[ j,j \right]} }} \le \beta \frac{{\det  {\mathcal{L_G}\left[ {j,k} \right]} }}{{\det  {\mathcal{L_G}\left[ j,j \right]} }} + \beta \frac{{\det  {\mathcal{L_G}\left[ {k,l} \right]} }}{{\det  {\mathcal{L_G}\left[ k,k \right]} }}.
	\end{equation}  
	In equation (\ref{eqn:corro1}), $\det  {\mathcal{L_G}\left[ k,k \right]} = \mathfrak{T}(\mathcal{G}), \forall k \in \mathcal{V_B}$. Therefore, 
	\begin{equation*}\label{eqn:corro2}
		\det\mathcal{L_G}\left[j,l\right] \le \det\mathcal{L_G}\left[j,k\right] + \det\mathcal{L_G}\left[k,l\right],	
	\end{equation*} 
	holds $\forall\, \Gamma \in \Gamma_{\beta}^{RL}$. Whereas pair $\left(\Re\left(t\right),\Im\left(t\right)\right) \in \mathcal{Q}_2$ for networks in $\Gamma_{\beta}^{RC}$, therefore
	\begin{equation*}\label{eqn:corro2}
		\det\mathcal{L_G}\left[j,l\right] \le \det\mathcal{L_G}\left[j,k\right] + \det\mathcal{L_G}\left[k,l\right],	
	\end{equation*} 
	holds $\forall\, \Gamma \in \Gamma_{\beta}^{RC}$. 
\end{proof}
\subsection{Kalmanson's Inequality on $\Upsilon$}\label{App:B}
\begin{center}
	\textcolor{black}{\textbf{Proof for Theorem-\ref{thm:kalman_ineq_RL_RC21}}}\\
\end{center}
\begin{proof}	
	Let $z_{j,l}^{th} + z_{k,m}^{th}$ be written as, 
	\begin{multline*}
		z_{j,l}^{th} + z_{k,m}^{th} = \left(e_{jk}+e_{kl}\right)^T \mathcal{L}_{\Gamma}^{\dagger}\left(e_{jk}+e_{kl}\right)\\ + \left(e_{kl}+e_{lm}\right)^T \mathcal{L}_{\Gamma}^{\dagger}\left(e_{kl}+e_{lm}\right),
	\end{multline*}	
	which gives,
	\begin{equation}\label{eqn:thm5_eqn2}
		\begin{split}
			z_{j,l}^{th} + z_{k,m}^{th} &= z_{j,k}^{th} + z_{l,m}^{th} + 2z_{k,l}^{th} + 2e_{jk}^T\mathcal{L}_{\Gamma}^{\dagger}e_{kl}+ \\& 2e_{kl}^T\mathcal{L}_{\Gamma}^{\dagger}e_{lm}=z_{j,k}^{th} + z_{l,m}^{th} + 2f_1.
		\end{split}
	\end{equation}
	Where, $f_1=z_{k,l}^{th} + e_{jk}^T\mathcal{L}_{\Gamma}^{\dagger}e_{kl}+e_{kl}^T\mathcal{L}_{\Gamma}^{\dagger}e_{lm}$ and $f \in \mathbb{C}$, which can be further written as,
	\begin{equation}\label{eqn:Mdef}
		\begin{split}
			f_1&=z_{k,l}^{th} + e_{jk}^T\mathcal{L}_{\Gamma}^{\dagger}e_{kl}+e_{kl}^T\mathcal{L}_{\Gamma}^{\dagger}e_{lm},\\
			&=e_{kl}^T\mathcal{L}_{\Gamma}^{\dagger}\left(e_{kl}+e_{jk}+e_{lm}\right)=e_{kl}^T\mathcal{L}_{\Gamma}^{\dagger}e_{jm}.
		\end{split}
	\end{equation}
	From equation (\ref{eqn:Mdef}), $f_1$ is a voltage across boundary nodes $k$ and $l$ when current $1\angle0$ Ampere is fed into node $j$ and node $m$ is grounded. Now,
	\begin{itemize}
		\item If $\Re(f_1)\ge0$ \& $\Im(f_1)\ge0,\,\mbox{then}$
		\begin{multline*}
			\left(r^{th}_{j,l},x^{th}_{j,l}\right)+\left(r^{th}_{k,m},x^{th}_{k,m}\right) \preceq^{\mathbf{r}} \left(r^{th}_{j,k},x^{th}_{j,k}\right)+\left(r^{th}_{l,m},x^{th}_{l,m}\right) \\ \implies z^{th}_{j,l}+z^{th}_{k,m} \preceq^{\mathbf{r}} z^{th}_{j,k}+z^{th}_{l,m}.	
		\end{multline*}
		\item If $\Re(f_1)\ge0$ \& $\Im(f_1)\le0$,\,then
		\begin{equation*}
			\left(r^{th}_{j,l}+r^{th}_{k,m},\,x^{th}_{j,k}+x^{th}_{l,m} \right)\preceq^{\mathbf{r}} \left(r^{th}_{j,k}+r^{th}_{l,m},\,x^{th}_{j,l}+x^{th}_{k,m}\right)	
		\end{equation*} 
		\item If $\Re(f_1)\le0$ \& $\Im(f_1)\ge0$,\,then
		\begin{equation*}
			\left(r^{th}_{j,k}+r^{th}_{l,m},\,x^{th}_{j,l}+x^{th}_{k,m} \right)\preceq^{\mathbf{r}} \left(r^{th}_{j,l}+r^{th}_{k,m},\,x^{th}_{j,k}+x^{th}_{l,m}\right) 	
		\end{equation*}
		\item If $\Re(f_1)\le0$ \& $\Im(f_1)\le0,\,\mbox{then}$
		\begin{multline*}
			\left(r^{th}_{j,l},x^{th}_{j,l}\right)+\left(r^{th}_{k,m},x^{th}_{k,m}\right) \preceq^{\mathbf{r}} \left(r^{th}_{j,k},x^{th}_{j,k}\right)+\left(r^{th}_{l,m},x^{th}_{l,m}\right) \\ \implies z^{th}_{j,k}+z^{th}_{l,m} \preceq^{\mathbf{r}} z^{th}_{j,l}+z^{th}_{k,m}.	
		\end{multline*}
	\end{itemize}
	Here $\mathbf{r}=1$ for RL network and $\mathbf{r}=2$ for RC network.\\
	Second set of inequalities are derived similarly starting from equation,
	\begin{multline*}
		z_{j,l}^{th} + z_{k,m}^{th} = \left(e_{jk}+e_{kl}\right)^T \mathcal{L}_{\Gamma}^{\dagger}\left(e_{jk}+e_{kl}\right)\\ + \left(e_{jk}+e_{jm}\right)^T \mathcal{L}_{\Gamma}^{\dagger}\left(e_{kj}+e_{jm}\right).
	\end{multline*}	
	Which can be further written as,
	\begin{equation}\label{eqn:thm5_eqn21}
		z_{j,l}^{th} + z_{k,m}^{th} = z^{th}_{k,l} + z^{th}_{j,m} + 2e_{jk}\mathcal{L}_{\Gamma}^{\dagger}e_{ml}.	
	\end{equation}	
\end{proof} 
\subsection{Kalmanson's Inequality on $\Gamma_{{\beta}}$}\label{App:B1}
\begin{center}
	\textcolor{black}{\textbf{Proof for Corollary-\ref{cor:cor8}}}\\
\end{center}
\begin{proof}
	We known from equation (\ref{eqn:thm5_eqn2}) that,
	\begin{equation*}
		z_{j,l}^{th} + z_{k,m}^{th} = z_{j,k}^{th} + z_{l,m}^{th} + 2z_{k,l}^{th} + 2e_{kl}^T\mathcal{L}_{\Gamma}^{\dagger}e_{jm}	
	\end{equation*} 
	Since $\Gamma \in \Gamma_{\beta}$, we write,
	\begin{equation}
		f_1=e_{kl}^T\mathcal{L}_{\Gamma}^{\dagger}e_{jm}=\gamma^{-1}e_{kl}^T\mathcal{L}_{\mathcal{G}}^{\dagger}e_{jm}.
	\end{equation}	
	Here, term $e_{kl}^T\mathcal{L}_{\mathcal{G}}^{\dagger}e_{jm}$ is the voltage across boundary nodes $k$ and $l$ when $1\angle0$ Ampere current is sent into $j^{th}$ node and $m^{th}$ node is grounded.

	\noindent For a network $\Gamma \in \Gamma_{\beta}^{RL}$, 
	\begin{itemize}
		\item if $e_{kl}^T\mathcal{L}_{\mathcal{G}}^{\dagger}e_{jm}>0$ then $f_1 \in \mathcal{Q}_1$. Therefore, $\Re \left( {z_{j,k}^{th}} \right) + \Re \left( {z_{l,m}^{th}} \right) \le \Re \left( {z_{j,l}^{th}} \right) + \Re \left( {z_{k,m}^{th}} \right).\,\textrm{Using relation in equation (\ref{imped_dist}), we get,}\\
		\det \left( {{\mathcal{L_G}}\left[ {j,k} \right]} \right) + \det \left( {{{\mathcal{L_G}}}\left[ {l,m} \right]} \right) \le\det \left( {{{\mathcal{L_G}}}\left[ {j,l} \right]} \right) + \\
		{\rm{                                                       }}\quad \quad \quad \quad \quad \quad \quad \quad \quad \quad \quad \quad \quad \quad \det \left( {{\mathcal{L_G}}\left[ {k,m} \right]} \right)$
		\item If $e_{kl}^T\mathcal{L}_{\mathcal{G}}^{\dagger}e_{jm}<0$ then $f_1 \in \mathcal{Q}_3$. Therefore, $\Re\left(z_{j,l}^{th}\right) + \Re\left(z_{k,m}^{th}\right) \le \Re\left(z_{j,k}^{th}\right) + \Re\left(z_{l,m}^{th}\right).\,\mbox{Using relation in equation (\ref{imped_dist}), we get,}\\ \det \left( {{\mathcal{L_G}}\left[ {j,l} \right]} \right) + \det \left( {{{\mathcal{L_G}}}\left[ {k,m} \right]} \right) \le\det \left( {{{\mathcal{L_G}}}\left[ {j,k} \right]} \right) + \\
		{\rm{                                                       }}\quad \quad \quad \quad \quad \quad \quad \quad \quad \quad \quad \quad \quad \quad \det \left( {{\mathcal{L_G}}\left[ {l,m} \right]} \right)$
	\end{itemize}
	Now, consider equation (\ref{eqn:thm5_eqn21}), 	
	\begin{equation*}
		z_{j,l}^{th} + z_{k,m}^{th} = z_{k,l}^{th} + z_{j,m}^{th} + 2e_{jk}^T\mathcal{L}_{\Gamma}^{\dagger}e_{ml} 
	\end{equation*}	
	we write,
	\begin{equation}
		f_2=e_{jk}^T\mathcal{L}_{\Gamma}^{\dagger}e_{ml}=\gamma^{-1}e_{jk}^T\mathcal{L}_{\mathcal{G}}^{\dagger}e_{ml}.
	\end{equation}	
	Here, term $e_{jk}^T\mathcal{L}_{\mathcal{G}}^{\dagger}e_{ml}$ is the voltage across boundary nodes $j$ and $k$ when $1\angle0$ Ampere current is sent into $m^{th}$ node and $l^{th}$ node is grounded.
	\begin{itemize}
		\item if $e_{jk}^T\mathcal{L}_{\mathcal{G}}^{\dagger}e_{ml}>0$ then $f_2 \in \mathcal{Q}_1 $. Therefore, $\Re \left( {z_{k,l}^{th}} \right) + \Re \left( {z_{j,m}^{th}} \right) \le \Re \left( {z_{j,l}^{th}} \right) + \Re \left( {z_{k,m}^{th}} \right)\,
		\implies \det \left( {{\mathcal{L_G}}\left[ {k,l} \right]} \right) + \det \left( {{{\mathcal{L_G}}}\left[ {j,m} \right]} \right) \le\det \left( {{{\mathcal{L_G}}}\left[ {j,l} \right]} \right) +  \det \left( {{\mathcal{L_G}}\left[ {k,m} \right]} \right)$
		\vspace{0.1cm}
		\item if $e_{jk}^T\mathcal{L}_{\mathcal{G}}^{\dagger}e_{ml}<0$ then $f_2 \in \mathcal{Q}_3 $. Therefore, $\Re \left( {z_{j,l}^{th}} \right) + \Re \left( {z_{k,m}^{th}} \right) \le \Re \left( {z_{k,l}^{th}} \right) + \Re \left( {z_{j,m}^{th}} \right)
		\implies \det \left( {{\mathcal{L_G}}\left[ {j,l} \right]} \right) + \det \left( {{{\mathcal{L_G}}}\left[ {k,m} \right]} \right) \le\det \left( {{{\mathcal{L_G}}}\left[ {k,l} \right]} \right) +  \det \left( {{\mathcal{L_G}}\left[ {j,m} \right]} \right).$
	\end{itemize} 
	Similarly, using equation (\ref{eqn:thm5_eqn2}) and equation (\ref{eqn:thm5_eqn21}), we can derive similar set of inequalities for any network $\Gamma \in \Gamma^{RC}_{\beta}$. 
\end{proof}
\subsection{Alternate Proof for Kalmanson's Inequality on $\Gamma^R_{{\beta}}$}\label{App:B2}
\begin{center}
	\textcolor{black}{\textbf{Proof for Corollary-\ref{triangle_gamma_beta}}}\\
\end{center}
\begin{proof}
	Consider	equation (\ref{eqn:thm5_eqn2}),
	\begin{equation*}
		z_{j,l}^{th} + z_{k,m}^{th} =z_{j,k}^{th} + z_{l,m}^{th} + 2e_{kl}^T\mathcal{L}_{\Gamma}^{\dagger}e_{jm}.
	\end{equation*}
	For $\Gamma \in \Gamma^{R}_{\beta}$, we write above equation as,
	\begin{equation}\label{eqn45}
		r_{j,l}^{th} + r_{k,m}^{th} =r_{j,k}^{th} + r_{l,m}^{th} + 2\beta e_{kl}^T\mathcal{L}_{\mathcal{G}}^{\dagger}e_{jm}.
	\end{equation}
	Let $f_1=e_{kl}^T\mathcal{L}_{\mathcal{G}}^{\dagger}e_{jm}$. To prove $f_1>0$, consider $\Lambda_{\mathcal{G}}\in \mathbb{R}^{4 \times 4}$ be the Kron reduced matrix corresponding to $\mathcal{L}_{\mathcal{G}}$. The Kron reduction is done with respect to nodes $j,k,l,m$. Therefore, $e_{kl}^T\mathcal{L}_{\mathcal{G}}^{\dagger}e_{jm}=e_{kl}^T\Lambda_{\mathcal{G}}^{\dagger}e_{jm}$. Since we assume that the network is connected, the Kron reduced network, labeled as $\mathcal{G}^{{kron}}$ is as given in Figure \ref{fig1:Gkron} below,
	\begin{figure}[H]
		\centering
		\includegraphics[scale=0.35]{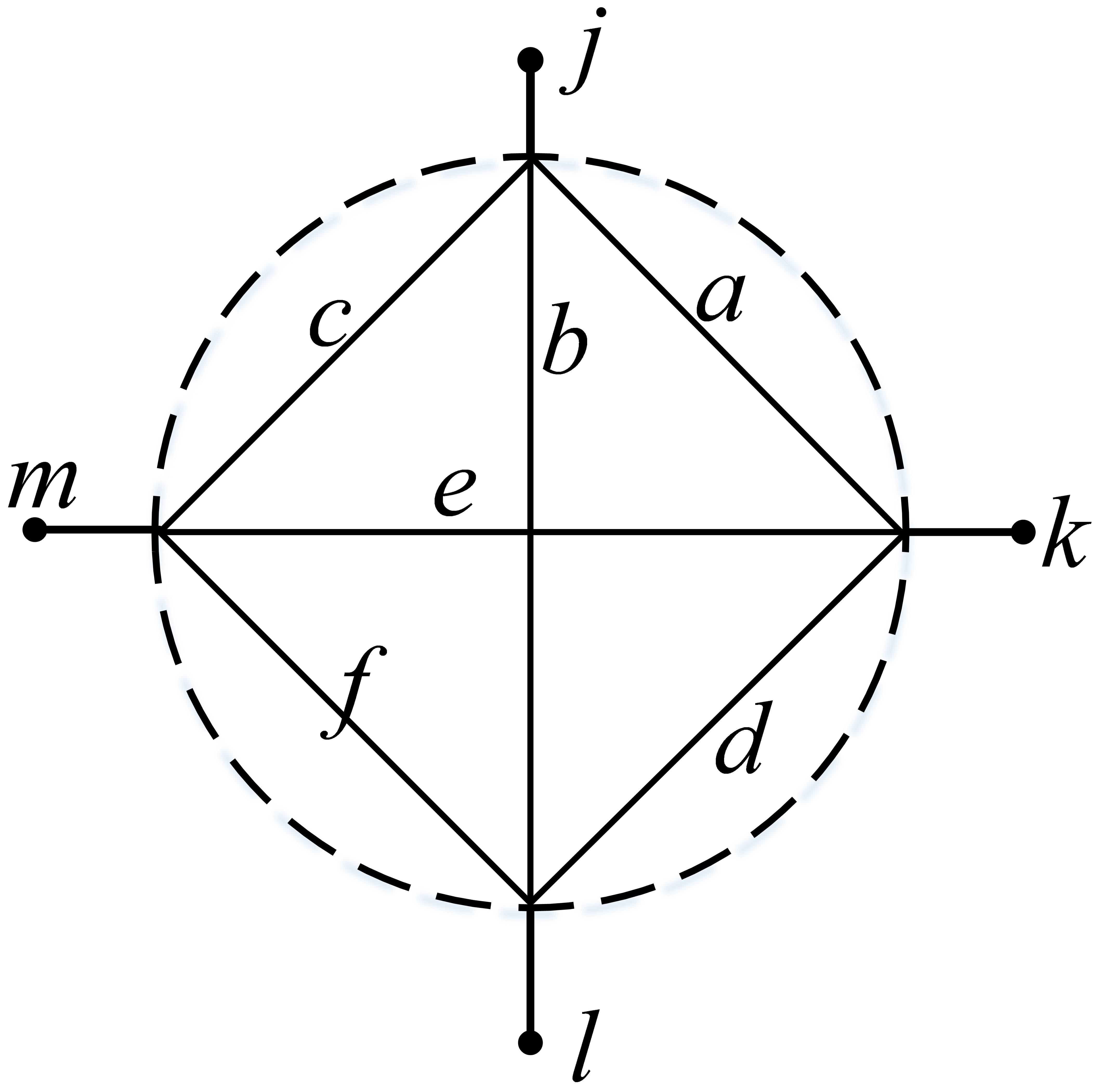}
		\caption{Kron reduced network $\mathcal{G}^{kron}$.}
		\label{fig1:Gkron}
	\end{figure}
	Here, $a,b,c,d,e,f$ are the edge resistances in $\mathcal{G}^{kron}$. The matrix $\Lambda_{\mathcal{G}}$ has the following general form,
	\begin{equation}
		{\Lambda  = \left[ {\begin{array}{*{20}{c}}
					\Lambda_{11}&{ - \frac{1}{{{a}}}}&{ - \frac{1}{{{b}}}}&{ - \frac{1}{{{c}}}}\\
					{- \frac{1}{{a}}}&\Lambda_{22}&{ - \frac{1}{{{d}}}}&{ - \frac{1}{{{e}}}}\\
					{ - \frac{1}{{{b}}}}&{ - \frac{1}{{{d}}}}&\Lambda_{33}&{ - \frac{1}{{{f}}}}\\
					{ - \frac{1}{{{c}}}}&{ - \frac{1}{{e}}}&{ - \frac{1}{{{f}}}}&\Lambda_{44}
			\end{array}} \right].}
	\end{equation}
	Therefore, the term $f_1$ computed using $\Lambda_{\mathcal{G}}$ is, 
	\begin{equation}
		f_1 = \frac{{cd\left( {be - af} \right)}}{\begin{array}{l}
				abd + abe + acd + abf + ace + bcd + acf + \\
				bce + adf + bcf + bde + aef + cde + bef + \\
				cdf + def
		\end{array}}.
	\end{equation}
	\noindent Now, let $\mathcal{P}=(k,l)$ and $\mathcal{Q}=(j,m)$ be two sequence of boundary nodes. Since, there exist two disjoint paths, one path connecting node $k$ to $j$ and, other path connecting $l$ to $m$, we then say $\mathcal{P}$ and $\mathcal{Q}$ are $2$-connected. Hence, from Theorem 3.13 in \cite{curtis2000inverse} the term $\left( {be - af} \right)>0$.
	\begin{equation}\label{rzineq}
		\left( {be - af} \right)>0 \Rightarrow f_1 > 0. 
	\end{equation} 
	Therefore, we write equation (\ref{eqn45}) as,
	\begin{equation}\label{eqn49}
		r_{j,l}^{th} + r_{k,m}^{th} \ge r_{j,k}^{th} + r_{l,m}^{th}.
	\end{equation}
	Using equation (\ref{imped_dist}) in equation (\ref{eqn45}), we get,
	$$\det\mathcal{L_G}\left[j,k\right] + \det\mathcal{L_G}\left[l,m\right] \le \det\mathcal{L_G}\left[j,l\right] + \det\mathcal{L_G}\left[k,m\right].$$

	\noindent Similarly other inequality is obtained, starting from equation (\ref{eqn:thm5_eqn21}). 
\end{proof}
\bibliographystyle{IEEEtran}
\bibliography{refer}

\end{document}